\documentclass[11pt]{article}
\usepackage[affil-it]{authblk}
\usepackage{amsmath}
\usepackage{amssymb}
\usepackage[a4paper]{geometry}
\usepackage{amsthm}
\usepackage{subcaption}
\usepackage{graphics,epsfig,color}
\usepackage{graphicx}
\usepackage{bm}
\usepackage{bbold}
\usepackage{slashed}
\usepackage[hidelinks]{hyperref}

\hypersetup{
    colorlinks = true,
    linkcolor = {black},
    citecolor = {blue},
    anchorcolor = {green},
    citecolor = {red},
    filecolor = {cyan},
    menucolor = {red},
    runcolor = {cyan},
    urlcolor = {magenta}
}
\newcommand{\twidth}{6in}
\textwidth=\twidth
\renewcommand{\AA}{{\mathcal{A}}}
\newcommand{\BB}{{\mathcal{B}}}
\newcommand{\CC}{{\mathcal{C}}}
\newcommand{\DD}{{\mathcal{D}}}
\newcommand{\EE}{{\mathcal{E}}}

\newcommand{\GG}{{\mathcal{G}}}

\newcommand{\MM}{{\mathcal{M}}}
\newcommand{\NN}{{\mathcal{N}}}

\newcommand{\QQ}{{\mathcal{Q}}}

\newcommand{\UU}{{\mathcal{U}}}

\newcommand{\su}{{\mathfrak{su}}}

\newcommand{\SU}{{\mathrm{SU}}}

\newcommand{\U}{{\mathrm{U}}}
\newcommand{\R}{{\mathbb{R}}}

\newcommand{\Z}{{\mathbb{Z}}}
\newcommand{\N}{{\mathbb{N}}}
\newcommand{\C}{{\mathbb{C}}}

\newcommand{\beq}{\begin{equation}}
\newcommand{\eeq}{\end{equation}}
\newcommand{\bea}{\begin{eqnarray}}
\newcommand{\eea}{\end{eqnarray}}
\newcommand{\bal}{\begin{align}}
\newcommand{\eal}{\end{align}}
\newcommand{\bml}{\begin{multline}}
\newcommand{\eml}{\end{multline}}
\newcommand{\bdy}{\partial}

\newcommand{\wt}{\widetilde}

\newcommand{\lto}{\longrightarrow}

\def \d{\mathrm{d}}
\newcommand{\ol}{\overline}
\newcommand{\tr}{{\rm tr}\, }

\newcommand{\ip}[1]{\langle#1\rangle}
\newcommand{\norm}[1]{\big\|#1\big\|}

\makeatletter
\newcommand\xleftrightarrow[2][]{%
  \ext@arrow 9999{\longleftrightarrowfill@}{#1}{#2}}
\newcommand\longleftrightarrowfill@{%
  \arrowfill@\leftarrow\relbar\rightarrow}
\makeatother

\newcommand{\bigslant}[2]{{\raisebox{.2em}{$#1$}\left/\raisebox{-.2em}{$#2$}\right.}}

\newtheorem{theorem}{Theorem}

\newtheorem{proposition}[theorem]{Proposition}

\definecolor{TW-color}{RGB}{100,0,100}

\numberwithin{equation}{section}
\begin{document}
%\maketitle
\renewcommand*{\thefootnote}{\fnsymbol{footnote}}
\begin{titlepage}
\begin{center}
{\huge
A model for gauged skyrmions\\
with low binding energies \par}

%\vspace{10mm}

\vspace{10mm}
{\Large Josh Cork\footnote{Email address: \texttt{joshua.cork@itp.uni-hannover.de}}$^{\rm 1}$,\ Derek Harland\footnote{Email address: \texttt{d.g.harland@leeds.ac.uk}}$^{\rm 2}$,\ and Thomas Winyard\footnote{Email address: \texttt{t.winyard@kent.ac.uk}}$^{\rm 2,3}$}\\[10mm]

\noindent {\em ${}^{\rm 1}$ Institut f\"ur Theoretische Physik \,{\rm and}
Riemann Center for Geometry and Physics\\
Leibniz Universit\"at Hannover, Appelstra\ss{}e 2, 30167 Hannover, Germany
}\\
\smallskip
\noindent {\em ${}^{\rm 2}$ School of Mathematics, University of Leeds\\ Woodhouse Lane, Leeds, United Kingdom }\\
\smallskip
\noindent {\em ${}^{\rm 3}$ School of Mathematics, Statistics, and Actuarial Science\\
University of Kent, Canterbury, United Kingdom}\\[10mm]
{\Large \today}
\vspace{15mm}
\begin{abstract}
    We consider gauged skyrmions with boundary conditions which break the gauge from $\SU(2)$ to $\U(1)$ in models derived from Yang--Mills theory. After deriving general topological energy bounds, we approximate charge $1$ energy minimisers using KvBLL calorons with non-trivial asymptotic holonomy, use them to calibrate the model to optimise the ratio of energy to lower bound, and compare them with solutions to full numerical simulation. Skyrmions from calorons with non-trivial asymptotic holonomy exhibit a non-zero magnetic dipole moment, which we calculate explicitly, and compare with experimental values for the proton and the neutron. We thus propose a way to develop a physically realistic Skyrme--Maxwell theory, with the potential for exhibiting low binding energies.
\end{abstract}
\end{center}
\end{titlepage}
\renewcommand*{\thefootnote}{\arabic{footnote}}
\setcounter{footnote}{0}
\hypersetup{
    linkcolor = {blue}
}
\section{Introduction}
The Skyrme model \cite{skyrme1962nucl} is a low-energy effective field theory for quantum chromodynamics (QCD). It was first proposed nearly sixty years ago as a nonlinear field theory for pions, however it only gained recognition nearly twenty years after its conception, when Witten demonstrated that baryons in the large $N$ limit of QCD \cite{witten1979baryons} may be identified as the topological soliton \cite{MantonSutcliffe2004,shnir2018topsolitons} solutions to the static field equations of the Skyrme model. These topological solitons are called \textit{skyrmions}. As an effective theory, the Skyrme model provides a tractable approach to studying nuclei, in comparison to trying to do so directly via QCD, moreover at both classical and quantum levels, the Skyrme model successfully predicts many observed properties of real nuclei with relative accuracy \cite{rhozahed2016multifaceted}.

The field content of the Skyrme model is an $\SU(2)$-valued map in $3+1$ dimensions called the Skyrme field, which encodes the pion degrees of freedom. It is possible to couple the Skyrme field to a gauge field, in which case one obtains the ingredients for a gauged Skyrme model. The traditional way to do this is to replace ordinary derivatives by covariant derivatives, and to include a Maxwell term in the Skyrme lagrangian. Gauged Skyrme models of this type have been proposed before for a variety of purposes: a $\U(1)$ model \cite{CallanWitten1984monopole} was introduced in order to study interactions between nucleons and monopoles, and an $\SU(2)_L$ model \cite{dHokerFarhi1984decoupling} was considered to study weak interactions in the Skyrme model (see also \cite{CriadoKhozeSpannowsky2020emergence}). Examples of $3$D static solutions to the $\U(1)$ model \cite{PietteTchrakian2000static,RaduTchrakian2006spinning} and $\SU(2)$ models \cite{ArthurTchrakian1996GaugedSky,BrihayeHartmannTchhrakian2001MonopolesGaugeSky} have been constructed numerically, although significantly less is known in general about gauged skyrmions when compared to their ungauged counterparts.

Both the ordinary and gauged Skyrme models possess topological invariants, often called the topological charge. In the ordinary model the charge is an integer, specifically the degree of the Skyrme field, and it is physically identified as the number of baryons. In addition, both theories exhibit topological energy bounds. Topological energy bounds, also known as Bogomol'nyi bounds, are ubiquitous in solitonic theories. These are lower bounds on the energy $E$ of a soliton, which are typically of the form
\begin{align}\label{topological-energy-bound-general}
E \geq C |\QQ|,
\end{align}
where $\QQ$ is the topological charge, and $C>0$ is a constant. Very often one is interested to know how close a soliton comes to saturating the bound, and a natural dimensionless measure of this is the quantity
\begin{align}\label{BMI}
E_{\text{BMI}}=\frac{E}{C|\QQ|}.
\end{align}
We will call this the \textit{Bogomol'nyi mass index} (or BMI for short).  Solitons with a BMI close to $1$ are close to saturating their topological energy bound, and solitons with a BMI of exactly $1$ are called BPS. Important examples of BPS solitons are vortices in gauged sigma models, Yang--Mills--Higgs monopoles, and Yang--Mills instantons.

Topological energy bounds and the Bogomol'nyi mass index are important in models of nuclei because they yield insight into nuclear binding energies.  The binding energy per nucleon of a skyrmion with topological charge $\QQ\in\N$ is defined to be $(\QQ E_1 - E_{\QQ})/\QQ$, where $E_{\QQ}$ and $E_1$ are the energies of the $\QQ$-skyrmion and the 1-skyrmion.  In nuclear physics binding energies per nucleon are around 1\% of the mass of a nucleon.  The energy bound \eqref{topological-energy-bound-general} implies that binding energies per nucleon are bounded above by $E_1-C$, and in particular, the binding energy as a percentage of the mass of a nucleon is bounded above by $E_{\mathrm{BMI},1}-1$, where $E_{\mathrm{BMI},1}$ is the BMI of the $1$-skyrmion. Therefore if the charge 1 skyrmion in some variant of the Skyrme model has a BMI a little above 1.01, the classical binding energies in that model are likely to be comparable with experimental values.

Due to geometrical constraints \cite{Manton1987geometry}, the standard Skyrme model does not admit BPS solutions, so skyrmions always exhibit BMIs strictly greater than $1$. Unfortunately, they are significantly greater than $1$: for example, the 1-skyrmion has a BMI of approximately $1.232$ \cite{MantonSutcliffe2004}.  As a result, binding energies in the standard Skyrme model can be very large (as much as 10\% of the mass of a nucleon \cite{battye2001solitonic,battye2002skyrmions}).  In gauged Skyrme models the situation can be even worse.  In the $\U(1)$-gauged Skyrme model studied in \cite{PietteTchrakian2000static} the 1-skyrmion at strong coupling has a BMI of approximately $1.820$, so the binding energy problem is likely exacerbated in this model.

There are currently a growing number of proposed ways to alter the standard Skyrme model to lower binding energies \cite{AdamSanchez-GuillenWereszczynski2010sexticmodel,sutcliffe2011skyrmionsTruncated,harland2014topological,GillardHarlandSpeight2015skyrmions,GillardHarlandKirkMaybeeSpeight2017,GudnasonSpeight2020realistic,AdamOlesWereszczynski2020dielectric,gudnason2020dielectric}. Perhaps one of the most extreme of these is Sutcliffe's holographic model \cite{sutcliffe2011skyrmionsTruncated}. The full model is a BPS theory which arises via an expansion of a Yang--Mills gauge field on $\R^4$ \cite{sutcliffe2010skyrmions} in terms of its holonomy and infinitely many vector mesons, and is the flat-space analogue of the holographic model of Sakai and Sugimoto \cite{SakaiSugimoto2005}. By truncating this expansion to only include the first vector meson -- identified as the $\rho$-meson -- one obtains a more manageable model, which has proven to be somewhat successful in reducing binding energies \cite{NayaSutcliffe2018skyrmions}, and appears to resolve other problems in the Skyrme model such as how it exhibits alpha-particle subclusters \cite{nayaSutcliffe2018skyrmionscluster}. On top of its phenomenological successes, a beautiful advantage of Sutcliffe's model is how it explains the accuracy of the Atiyah--Manton approximation \cite{AtiyahManton1989} of Skyrme fields from Yang--Mills instantons; the holonomy of instantons on $\R^4$ well-approximate skyrmions, and Sutcliffe's model provides a solid theoretical framework for controlling how good the approximation is. This relationship with instantons also allows for a consistent configuration space on which to study low-energy interactions of skyrmions \cite{AtiyahManton1993geometry,HalcrowWinyard2021consistent}.

As demonstrated by one of us, the Atiyah--Manton--Sutcliffe framework can be extended to gauged skyrmions \cite{cork2018skyrmions}, specifically, providing a formulation of an $\SU(2)$ gauged Skyrme model on $\R^3$, where gauged skyrmions are seen to be well-approximated by self-dual Yang--Mills fields on $S^1\times\R^3$, namely calorons. Even after discarding all vector mesons, the gauged Skyrme model which arises in this way is more general than the more traditional gauged Skyrme models, with the addition of interaction terms between the Skyrme and gauge field. An advantage of these more general models is that there is additional freedom in the choice of parameters, and one aim of this paper is to demonstrate how to fine-tune these parameters so as to optimise the BMI of gauged skyrmions.

Following on from the initial work in \cite{cork2018skyrmions} for spherically-symmetric configurations, in this paper we consider charge $|\QQ|=1$ gauged skyrmions derived from calorons with non-trivial holonomy \cite{KraanVanBaal1998,LeeLu1998} which posses axial symmetry. The unique feature that we are interested in is that their boundary conditions break the gauge symmetry from $\SU(2)$ to $\U(1)$, and this passes over to the associated gauged skyrmions. So whilst our model is generically $\SU(2)$, the existence of this boundary condition means such skyrmions can be used as a toy model for studying $\U(1)$ gauged skyrmions.

This paper is organised as follows. In section \ref{sec:gauged-skyrme-models} we introduce a general prescription for an $\SU(2)$ gauged Skyrme model, alongside a review of its construction from Yang--Mills theory, and its associated topological charge. In section \ref{sec:energy-bounds} we derive topological lower bounds for the gauged Skyrme energy. Two distinct bounds are considered: the first is a general bound derived for couplings defined in the interior of a domain $\DD\subset\R^6$, and the second of these considers behaviour on the boundary of $\DD$. A comparison to the Yang--Mills topological bound is made. Section \ref{sec:opt-BMI} is dedicated to computing the energies of gauged Skyrme fields derived from calorons with non-trivial holonomy, and comparing to numerical minimisers, with an emphasis on optimising the BMI with respect to the bounds from section \ref{sec:energy-bounds}. In particular, we show that the model with optimal BMI for these configurations must be defined by couplings on the boundary of $\DD$. Finally in section \ref{sec:dipole}, to further emphasise the relationship to $\U(1)$ gauged skyrmions, we exploit the fact that calorons with non-trivial holonomy may be viewed as consisting of two constituent monopoles \cite{kraanVanBaal1998monopoleconsts}, and we compute the associated magnetic dipole moment. Since our work relies heavily on a relatively detailed understanding of calorons with non-trivial holonomy, we provide a review of the charge $1$ examples which we consider in appendix \ref{appendix:KvBLL}.
\section{Gauged Skyrme models}\label{sec:gauged-skyrme-models}

\subsection{The energy functional}
Let $U:\R^3\lto \SU(2)$ be a smooth function, and $A$ be an $\SU(2)$ gauge field. $A$ defines a covariant derivative $D^A$, which acts on $U$ via $D^AU=\d U+[A,U]$. This defines in turn the $\su(2)$-valued currents $L^A=U^{-1}D^AU$ and $R^A=D^AUU^{-1}$, and the non-abelian curvature $F^A=\d A+A\wedge A$. We consider the following energy functional for static gauged Skyrme fields:
\begin{multline}
    E[U,A]=\chi_0\norm{L^A}^2+\chi_1\norm{L^A\wedge L^A}^2+\chi_2\norm{F^A}^2\\-\chi_3\ip{F^A,U^{-1}F^AU}-\ip{\chi_4F^A+\chi_5U^{-1}F^AU,L^A\wedge L^A}.\label{Gauged-Skyrme-energy}
\end{multline}
Here $\chi_p$ are in general arbitrary real parameters, which we shall constrain later, and the inner product $\ip{\cdot,\cdot}$ and norm $\norm{\cdot}$ are defined on $p$-forms $\xi,\eta\in\Omega^p(\R^3,\su(2))$ as
\begin{align}
    \ip{\xi,\eta}\equiv\ip{\xi,\eta}_{L^2}=\int_{\R^3}\tr\left(\xi\wedge\star\eta^\dagger\right),\quad \norm{\xi}^2=\int_{\R^3}\tr\left(\xi\wedge\star\xi^\dagger\right),\label{ip-forms}
\end{align}
where $\star:\Lambda^p\lto\Lambda^{3-p}$ is the Hodge isomorphism, and $\eta^\dagger$ denotes the $p$-form with components $\eta_{a_1\cdots a_p}^\dagger$.

The energy \eqref{Gauged-Skyrme-energy} is invariant under gauge transformations $g:\R^3\lto \SU(2)$, which act on the field configurations $(U,A)$ via
\begin{align}\label{gauge-transformations}
    U\mapsto g^{-1}Ug,\quad A\mapsto g^{-1}Ag+g^{-1}\d g.
\end{align}
It is also invariant under Euclidean transformations of $\R^3$, including parity-reversing transformations.  If $\chi_4=\chi_5$ it becomes invariant under the discrete symmetry
\begin{align}\label{discrete-symmetry}
(U,A)\mapsto(U^T,-A^T),
\end{align}
which is related to charge conjugation.  The static energy \eqref{Gauged-Skyrme-energy} can easily be extended to a Lorentz-covariant action, although the focus of the present article is on static solutions only.  For applications in physics, one is ultimately interested in a Skyrme model coupled to a $\U(1)$ gauged field.  Although the energy \eqref{Gauged-Skyrme-energy} involves instead an $\SU(2)$ gauge field, we will show later that the gauge group can be broken to $\U(1)$ by a suitable choice of boundary condition.

Several earlier papers on gauged skyrmions considered only a restricted model in which $\chi_3=\chi_4=\chi_5=0$ \cite{PietteTchrakian2000static,ArthurTchrakian1996GaugedSky, BrihayeHartmannTchhrakian2001MonopolesGaugeSky,brihayeTchhrakian1998solitons}.  However, the more general model \eqref{Gauged-Skyrme-energy} seems natural in the sense that it is consistent with expected symmetries.  Moreover, as will be shown below, models of the form \eqref{Gauged-Skyrme-energy} arise naturally from an holographic construction.

It will prove convenient to rewrite the energy \eqref{Gauged-Skyrme-energy} slightly.  Let $F_\pm^{U,A} = F^A \pm U^{-1}F^A U$ and let
\begin{align}\label{chi-to-x-couplings}
\begin{aligned}
x_1 &= \chi_0, & x_3 &= \chi_4+\chi_5, & x_5 &= \chi_2+\chi_3, \\
x_2 &= \chi_1, & x_4 &= \chi_4-\chi_5, & x_6 &= \chi_2-\chi_3.
\end{aligned}
\end{align}
With these parameters, the energy \eqref{Gauged-Skyrme-energy} may be written as
\begin{multline}
    E[U,A]=x_1\norm{L^A}^2+x_2\norm{L^A\wedge L^A}^2+\dfrac{x_5}{4}\norm{F_-^{U,A}}^2+\dfrac{x_6}{4}\norm{F_+^{U,A}}^2\\-\dfrac{1}{2}\ip{x_3F_+^{U,A}+x_4F_-^{U,A},L^A\wedge L^A},\label{Gauged-Skyrme-energy-x-vars}
\end{multline}
and the additional discrete symmetry \eqref{discrete-symmetry} arises when $x_4=0$.  This choice of parameters simplifies many of the calculations that follow.

It is sometimes useful to rewrite the energy \eqref{Gauged-Skyrme-energy-x-vars} in terms of a 3-sphere valued map.  We write
\begin{align}
    U=\phi^0\mathbb{1}+{\rm i}\vec{\phi}\cdot\vec{\sigma},
\end{align}
with $\phi=(\phi^0,\vec{\phi})\equiv(\phi^0,\phi^1,\phi^2,\phi^3)$ a unit four-vector and $\vec{\sigma}$ the vector of Pauli matrices.  Additionally, we may think of the gauge field $A$ as a three-vector $(A^1,A^2,A^3)$ of one-forms, namely $A={\rm i} A_i^a\sigma^a\;\d X_i$.\footnote{Here, and throughout, we denote a point in $\R^3$ by $\vec{X}=(X_1,X_2,X_3)$ so as to distinguish it from the couplings $x_p$ \eqref{chi-to-x-couplings} for the energy \eqref{Gauged-Skyrme-energy-x-vars}.} The covariant derivative then acts on the components $\phi^\mu$ as
\begin{align}
    D_i\phi^0=\bdy_i\phi^0,\quad D_i\phi^a=\bdy_i\phi^a-2\epsilon_{abc}A_i^b\phi^c.
\end{align}
The curvature may also be written as $F^A=\frac{{\rm i}}{2}F_{ij}^a\sigma^a\;\d X_i\wedge\d X_j$, where
\begin{align}
    F_{ij}^a=\bdy_iA_j^a-\bdy_jA_i^a-2\epsilon_{abc}A_i^bA_j^c.
\end{align}
Summation over repeated indices is assumed here. In this way, the energy \eqref{Gauged-Skyrme-energy-x-vars} may be written as
\begin{align}
    E=\int_{\R^3}\EE\;\d^3x,
    \label{expanded-energy}
\end{align}
with the energy density $\EE$ in tensor components given by
\begin{multline}
    \EE=2x_1\big|D_i\phi^\mu\big|^2+4x_2\big|D_{[i}\phi^\mu D_{j]}\phi^\nu\big|^2+x_6\big|F_{ij}^a\big|^2+2(x_5-x_6)\big|\phi^{[a}F_{ij}^{b]}\big|^2\\
    +4x_4F_{ij}^aD_{[i}\phi^0D_{j]}\phi^a+2x_3\epsilon_{abc}F_{ij}^aD_{[i}\phi^bD_{j]}\phi^c,
\end{multline}
where here Latin indices sum over $1,2,3$, Greek indices sum over $0,1,2,3$, and square brackets indicate antisymmetrised indices (e.g.\ $\phi^{[a}F_{ij}^{b]}=\frac12(\phi^{a}F_{ij}^{b}-\phi^{b}F_{ij}^{a})$).
\subsection{Gauged model from Yang--Mills theory}\label{subsec:sky-from-YM}
It was shown in \cite{cork2018skyrmions} that gauged Skyrme models of the form \eqref{Gauged-Skyrme-energy} can be obtained by taking a mode expansion of Yang--Mills theory on $S^1\times\R^3$.  We now review this construction.

From an $\SU(2)$ gauge field $\wt{A}$ on $\R\times\R^3$ with period $\beta$, that is, satisfying
\begin{align}
\wt{A}_\mu(t-\beta/2,\vec{X})=\wt{A}_\mu(t+\beta/2,\vec{X}),\quad\text{for all }t\in\R,
\end{align}
we may produce a Skyrme field $U:\R^3\lto \SU(2)$ from the holonomy around the periodic direction, namely, by solving for $H:[-\beta/2,\beta/2]\times\R^3\lto \SU(2)$:
\begin{align}\label{holonomy-cal-sky}
    \bdy_tH+\wt{A}_tH=0,\quad H(-\beta/2,\vec{X})=\mathbb{1},
\end{align}
and setting $U(\vec{X})=H(\beta/2,\vec{X})$. We also obtain a gauge field $A$ on $\R^3$ by setting
\begin{align}
A_j(\vec{X})=\wt{A}_j(-\beta/2,\vec{X})\quad\text{for }j=1,2,3.
\end{align}
Gauge-transforming with $H^{-1}$, one obtains a gauge where $\wt{A}_t=0$, which breaks the periodicity, and in this gauge, one can expand the fields $\wt{A}_j$ in the $t$-direction in terms of a complete, orthonormal basis of functions on $L^2[-\beta/2,\beta/2]$. By truncating this expansion to retain only leading terms, one is left with a gauge field of the form
\begin{align}
    \wt{A}_j(t,\vec{X})=\varphi_+(t)U(\vec{X})^{-1}D^A_jU(\vec{X})+A_j(\vec{X}),\label{expanded-gf}
\end{align}
where $\varphi_+$ is a `kink function' satisfying $\varphi_+(-\beta/2)=0$ and $\varphi_+(\beta/2)=1$. This may be defined, for example, by a normalised integral of one of the basis elements of $L^2[-\beta/2,\beta/2]$. Substituting \eqref{expanded-gf} into the Yang--Mills energy $\int_{S^1\times\mathbb{R}^3}\tr(F^{\wt{A}}\wedge\star(F^{\wt{A}})^\dagger)$ we obtain the static energy for a gauged Skyrme model. It may be described by \eqref{Gauged-Skyrme-energy-x-vars} with coefficients
\begin{align}\label{YM-couplings-ints}
\begin{aligned}
x_1&=\int_{-\frac{\beta}{2}}^{\frac{\beta}{2}}\left(\dfrac{\d\varphi_+}{\d t}\right)^2\;\d t, & 
x_2&=\int_{-\frac{\beta}{2}}^{\frac{\beta}{2}}\left(1-\varphi_+\right)^2\varphi_+^2\;\d t, \\
x_3&=2\int_{-\frac{\beta}{2}}^{\frac{\beta}{2}}\varphi_+\left(1-\varphi_+\right) \;\d t, & 
x_4&=2\int_{-\frac{\beta}{2}}^{\frac{\beta}{2}}\varphi_+\left(1-\varphi_+\right)\left(1-2\varphi_+\right) \;\d t, \\
x_5&=\int_{-\frac{\beta}{2}}^{\frac{\beta}{2}}\left(1-2\varphi_+\right)^2 \;\d t, & 
x_6&=\beta. 
\end{aligned}
\end{align}
Note that for any choice of the kink function $\varphi_+\in L^2[-\beta/2,\beta/2]$, we can rescale and consider instead the function $\wt{\varphi}_+\in L^2[-1/2,1/2]$ given by $\wt{\varphi}_+(t)=\varphi_+(t\beta)$. This rescales the coefficients \eqref{YM-couplings-ints} via $x_\mu\mapsto\wt{x}_\mu$ with $\wt{x}_1=\beta x_1$, and $\wt{x}_\mu=x_\mu/\beta$ for $\mu=2,\dots,6$. This is equivalent to a rescaling of length units, and thus we may always (when it is convenient to do so) set the period as $\beta=1$.

At first glance, this is a six-parameter family of models, but in reality there are a maximum of five independent parameters defined by \eqref{YM-couplings-ints}, meaning the model from Yang--Mills theory is a restricted case of the more general model \eqref{Gauged-Skyrme-energy-x-vars}. Indeed, note that the parameters \eqref{YM-couplings-ints} satisfy the relationship
\begin{align}\label{x-relationships-YM}
x_5+2x_3 = \beta = x_6
\end{align}
so that one of the variables ($x_5$, say) is determined by the other five. Furthermore, the remaining degrees of freedom in \eqref{YM-couplings-ints} are \textit{generically independent}, by which we mean that their integrands are linearly independent polynomials in $\varphi_+$ and $\varphi_+'$.

In principle, any function $\varphi_+$ with the correct boundary conditions could be chosen to define the couplings via \eqref{YM-couplings-ints}. For most of this paper, we follow \cite{cork2018skyrmions} and make the choice $\varphi_+=\phi_+^\alpha$ which arises from the orthonormal basis for $L^2[-\beta/2,\beta/2]$ given by the \textit{ultraspherical functions}\footnote{also known as the \textit{Gegenbauer} functions.}.  Explicitly,
\begin{align}
        \phi_+^\alpha(t)&=\bigslant{\int_{-\frac{\beta}{2}}^t\phi_0^\alpha(s)\;\d s}{\int_{-\frac{\beta}{2}}^{\frac{\beta}{2}}\phi_0^\alpha(s)\;\d s}\notag\\
       &=\dfrac{1}{2}+\dfrac{t}{2}\dfrac{4^{\alpha+1}\Gamma(\alpha+1)^2}{\Gamma(2\alpha+2)}\;{}_2F_1\left(\dfrac{1}{2},-\alpha;\dfrac{3}{2};\left(\frac{2t}{\beta}\right)^2\right),\label{hypergeometric-fct}
\end{align}
where
\begin{align}
       \phi_0^\alpha(t)=\left(1-\left(\dfrac{2t}{\beta}\right)^2\right)^\alpha,
\end{align}
$\alpha>-\frac{1}{2}$ is a real parameter, $\Gamma$ is the usual `gamma' function, and ${}_2F_1$ is Gauss' hypergeometric function.\footnote{See \cite{AbramowitzStegun1964} for details about these special functions. Our notation for the ultraspherical parameter differs slightly by replacing $\alpha$ by $2\alpha+\frac{1}{2}$, in agreement with \cite{cork2018skyrmions}.} For all $\alpha$, \eqref{hypergeometric-fct} satisfies the parity condition
\begin{align}
    \phi_+^\alpha(t)+\phi_+^\alpha(-t)=1,
\end{align}
which forces the integrand of $x_4$ in \eqref{YM-couplings-ints} to be odd, so that $x_4=0$, and hence the energy generated by these functions has the symmetry \eqref{discrete-symmetry}. In particular, the family \eqref{hypergeometric-fct} includes the simplest possible function satisfying the given boundary conditions, namely
\begin{align}\label{alpha=0-simple}
    \phi_+^0(t)=\dfrac{1}{2}+\dfrac{t}{\beta}.
\end{align}
The main motivation for choosing this family is due to how they generalise Sutcliffe's model \cite{sutcliffe2011skyrmionsTruncated,sutcliffe2010skyrmions}. In the weak-coupling limit where $\alpha,\beta\to\infty$ and $\beta^2/\alpha\to8$, the ultraspherical functions approach the Hermite functions, and in particular,
\begin{align}
    \phi_+^\alpha(t)\to\psi_+(t)=\dfrac{1}{2}\left(1+\mathrm{erf}\left(\frac{t}{\sqrt{2}}\right)\right),
\end{align}
which is precisely the kink function on $\R$ considered in \cite{sutcliffe2010skyrmions}.

The function $\phi_+^\alpha$ in \eqref{hypergeometric-fct} has a more transparent representation in the case where $\alpha\in\Z^+$, which we now derive. First, for brevity, let $N=\int_{-\frac{\beta}{2}}^{\frac{\beta}{2}}\phi_0^\alpha(t)\;\d t$. Then, integrating by parts $\alpha$ times gives
\begin{align*}
\phi_+^\alpha(t)
%&= \frac{1}{N}\int_{-\frac{\beta}{2}}^{t}\left(1+\frac{2s}{\beta}\right)^{\alpha}\left(1-\frac{2s}{\beta}\right)^{\alpha}\;\d s \\
% &= \frac{1}{N}\frac{\alpha}{\alpha+1}\int_{-\frac{\beta}{2}}^{t}\left(1+\frac{2s}{\beta}\right)^{\alpha+1}\left(1-\frac{2s}{\beta}\right)^{\alpha-1}\;\d s \nonumber \\
% &\quad+\frac{\beta}{2N}\frac{1}{\alpha+1}\left(1+\frac{2t}{\beta}\right)^{\alpha+1}\left(1-\frac{2t}{\beta}\right)^{\alpha} \\
% &= \ldots \\
&=\frac{1}{N}\frac{(\alpha!)^2}{(2\alpha)!}\int_{-\frac{\beta}{2}}^{t}\left(1+\frac{2s}{\beta}\right)^{2\alpha}\;\d s \nonumber\\
&\quad + \frac{\beta}{2N}\sum_{p=1}^{\alpha}\frac{(\alpha!)^2}{p!(2\alpha+1-p)!}\left(1+\frac{2t}{\beta}\right)^{2\alpha+1-p}\left(1-\frac{2t}{\beta}\right)^{p}\\
&= \frac{\beta}{2N}\frac{(\alpha!)^2}{(2\alpha+1)!}\sum_{p=0}^{\alpha}\left(\begin{array}{c}2\alpha+1\\p\end{array}\right)\left(1+\frac{2t}{\beta}\right)^{2\alpha+1-p}\left(1-\frac{2t}{\beta}\right)^{p}.
\end{align*}
Using
\begin{align}
N = 4^\alpha\beta\frac{(\alpha!)^2}{(2\alpha+1)!},
\end{align}
this gives the formula
\begin{align}\label{binomial-phi}
\phi_+^\alpha(t) = \sum_{p=0}^{\alpha}\left(\begin{array}{c}2\alpha+1\\p\end{array}\right)\left(\frac{1}{2}+\frac{t}{\beta}\right)^{2\alpha+1-p}\left(\frac{1}{2}-\frac{t}{\beta}\right)^{p}.
\end{align}
This simple formula shows that $\phi_+^\alpha(t)$ consists of the first $\alpha+1$ terms in the binomial expansion of $((\frac{1}{2}+\frac{t}{\beta})+(\frac{1}{2}-\frac{t}{\beta}))^{2\alpha+1}$ in powers of $\frac{1}{2}\pm\frac{t}{\beta}$.  The sum of the remaining terms is $\phi_+^\alpha(-t)$, so this expression makes it clear that $\phi_+^\alpha(t)+\phi_+^\alpha(-t)=1$.

With $\varphi_+$ chosen as in \eqref{hypergeometric-fct} and $\alpha\in(-\frac{1}{2},\infty)$, the coefficients in \eqref{Gauged-Skyrme-energy-x-vars} reduce to
\begin{align}\label{Ultra-spherical-couplings-1}
\begin{aligned}
x_1 &= \kappa_0 & x_3 &= \kappa_2 & x_5 &= \beta-2\kappa_2 \\
x_2 &= \frac{\kappa_1}{2} & x_4 &= 0 & x_6 &= \beta,
\end{aligned}
\end{align}
where $\kappa_0,\kappa_1$, and $\kappa_2$ are given by
\begin{align}
\kappa_0(\alpha,\beta)&=\frac{1}{\beta}\frac{\Gamma(2\alpha+1)^2\Gamma(2\alpha+2)^2}{\Gamma(4\alpha+2)\Gamma(\alpha+1)^4},\label{k_0}\\
\kappa_1(\alpha,\beta)&=\beta+2I_4(\alpha,\beta)-4I_2(\alpha,\beta),\label{k_1}\\
\kappa_2(\alpha,\beta)&=\beta-2I_2(\alpha,\beta)\label{k_2}
\end{align}
and
\begin{align}
    I_r(\alpha,\beta)=\int_{-\frac{\beta}{2}}^{\frac{\beta}{2}}\left(\phi_+^\alpha(t)\right)^r\;\d t.
\end{align}
One can show \cite{cork2018skyrmions} that
\begin{align}\label{I_2}
    I_2(\alpha,\beta)=\beta\left(\dfrac{1}{2}-\dfrac{\Gamma(2\alpha+2)^4}{\Gamma(4\alpha+4)\Gamma(\alpha+1)^3\Gamma(\alpha+2)}\right),
\end{align}
but no such closed-form expression has been found for $I_4(\alpha,\beta)$. However, $I_4$ can in principle be calculated explicitly for any fixed $\alpha\in\Z^+$, since in those cases \eqref{hypergeometric-fct} becomes a polynomial with rational coefficients, namely \eqref{binomial-phi}.
\subsection{Topological charge}
Ordinary, ungauged Skyrme fields are constant at spatial infinity, and thus have a well-defined integer degree $\BB$, which is interpreted physically as the baryon number. This may be calculated via
\begin{align}\label{baryon-number}
    \BB=\dfrac{1}{24\pi^2}\int_{\R^3}\tr(L\wedge L\wedge L),
\end{align}
with $L=U^{-1}\d U$ the ungauged current. When a gauging prescription is introduced, this quantity is no longer gauge-invariant, so a new definition of topological charge is required. Unfortunately, the na\"ive choice of replacing ordinary derivatives by covariant derivatives in \eqref{baryon-number} is no longer a topological invariant. The correct choice for the topological charge is \cite{CallanWitten1984monopole,PietteTchrakian2000static,ArthurTchrakian1996GaugedSky,BrihayeHartmannTchhrakian2001MonopolesGaugeSky,brihayeTchhrakian1998solitons}
\begin{align}\label{Gauged-Skyrme-charge}
    \QQ=\dfrac{1}{8\pi^2}\int\tr\left(\dfrac{1}{3}L^A\wedge L^A\wedge L^A-L^A\wedge\left(F^A+U^{-1}F^AU\right)\right).
\end{align}
If the Skyrme field is obtained from a gauge field $\wt{A}$ on $S^1\times\R^3$, as discussed in section \ref{subsec:sky-from-YM}, one can show that \cite{cork2018skyrmions} this coincides with the Yang--Mills topological charge,
\begin{align}\label{YM-charge}
    Q=\dfrac{1}{8\pi^2}\int\tr\left(F^{\wt{A}}\wedge F^{\wt{A}}\right).
\end{align}
The topological charge $\QQ$ is invariant under continuous deformations that fix the boundary conditions.  This follows immediately from its relationship with the Yang--Mills charge \eqref{YM-charge}, but can be seen more directly by the following argument, which mirrors that of \cite{ArthurTchrakian1996GaugedSky}.  Let
\begin{align}
\begin{aligned}
    \omega_0&=\frac{1}{3}\tr\left(L\wedge L\wedge L\right),&
    \rho&=\tr\left(\dfrac{1}{3}L^A\wedge L^A\wedge L^A-L^A\wedge\left(F^A+U^{-1}F^AU\right)\right),
\end{aligned}
\end{align}
denote the integrands of \eqref{baryon-number} and \eqref{Gauged-Skyrme-charge} respectively. One can show that (locally) we have
\begin{align}
    \rho-\omega_0
    =&\;\d\tr\left(\d_AU\wedge AU^{-1}-\d_AU^{-1}\wedge AU\right)\equiv\d\Omega,
\end{align}
where we have introduced the notation
\begin{align}
    \d_AU=\d U+\dfrac{1}{2}AU.
\end{align}
Thus
\begin{align}\label{charge-baryon-relationship}
    \QQ=\BB+\dfrac{1}{8\pi^2}\int_{\R^3}\d\Omega.
\end{align}
What this shows is that $\QQ$ is a standard topological invariant, plus a term which, by Stokes' theorem, only depends on the boundary conditions of the fields.

For the purpose of the following section, it is worth noting that the topological charge may be written in terms of the inner product \eqref{ip-forms} as
\begin{align}
    \QQ=\dfrac{1}{8\pi^2}\left\langle \star L^A,F^A+U^{-1}F^AU-\dfrac{1}{3}L^A\wedge L^A\right\rangle.\label{top-charge}
\end{align}
\section{Topological energy bounds}\label{sec:energy-bounds}

In this section we discuss topological energy bounds for the gauged Skyrme model, i.e.\ lower bounds on the energy \eqref{Gauged-Skyrme-energy-x-vars} which depend only on the charge \eqref{top-charge}.  We do so by applying a strategy developed in \cite{cork2018skyrmions} to this more general energy functional.  The bound that we derive in Theorem \ref{thm:energy-bound} is valid for parameters $(x_1,\ldots, x_6)$ lying in an open domain in $\R^6$.  We go on to derive a specialised bound in Theorem \ref{thm:restricted-energy-bound} which is valid on the boundary of this domain and which will be needed later.  Finally, we specialise to Skyrme models derived from Yang--Mills theory and compare our bounds with the Yang--Mills bound.
\subsection{General bounds for gauged Skyrme models}\label{subsec:gen-bounds}
To derive bounds for the energy \eqref{Gauged-Skyrme-energy-x-vars}, first we associate the two-forms $\star L^A$, $L^A\wedge L^A$, $F^A+U^{-1}F^AU$, and $F^A-U^{-1}F^AU$ with the variables $x,y,u,$ and $v$ respectively. Then the energy \eqref{Gauged-Skyrme-energy-x-vars} and charge \eqref{top-charge} are naturally associated with the following quadratic forms on $\R^4$:
\begin{align}
    \Omega_E(x,y,u,v)&=x_1x^2+x_2y^2+\dfrac{x_6}{4}u^2+\dfrac{x_5}{4}v^2-\left(\dfrac{x_3}{2}u+\dfrac{x_4}{2}v\right)y,\label{Omega_E}\\
    \Omega_{\QQ}(x,y,u,v)&=xu-\dfrac{1}{3}xy.\label{Omega_Q}
\end{align}
Since the bilinear form defined in \eqref{ip-forms} is positive definite, any linear combination of the energy \eqref{Gauged-Skyrme-energy-x-vars} and charge \eqref{top-charge} will be positive definite whenever the corresponding linear combination $a\Omega_E+b\Omega_{\QQ}$ is positive definite.

Considering these quadratic forms is sufficient for determining a topological bound of the form
\begin{align}\label{theoretical-bound}
    E\geq 8\pi^2C|\QQ|.
\end{align}
We shall do this by finding the maximal value of $C$ such that
\begin{align}
    \Omega(C)=\Omega_E\pm C\Omega_\QQ
\end{align}
is positive semidefinite.  For this strategy to work we need the quadratic form $\Omega_E$ to be positive definite.

\begin{proposition}\label{prop:energy-pos-def}
The quadratic form $\Omega_E$ is positive definite if and only if $x_1,x_2>0$, alongside the conditions
\begin{itemize}
    \item[(i)] $4x_2x_6>x_3^2$;
    \item[(ii)] $(4x_2x_6-x_3^2)x_5>x_6x_4^2$.
\end{itemize}
Hence the energy \eqref{Gauged-Skyrme-energy-x-vars} is positive if these conditions hold.
\end{proposition}
\begin{proof}
The associated symmetric matrix $M_E$ for the quadratic form $\Omega_E$ has leading principal minors
\begin{align}
\begin{aligned}
    m_1&=x_1,&
    m_3&=\dfrac{x_1}{16}\left(4x_2x_6-x_3^2\right),\\
    m_2&=x_1x_2,&
    m_4&=\dfrac{x_1}{64}\left((4x_2x_6-x_3^2)x_5-x_6x_4^2\right).
\end{aligned}
\end{align}
By Sylvester's criterion, $M_E$ (and hence $\Omega_E$) is positive definite if and only if these are all strictly positive, which is equivalent to the conditions stated.
\end{proof}

We shall denote the subset of $\R^6$ defined by the conditions of Proposition \ref{prop:energy-pos-def} by $\DD$. We now derive a lower bound in the cases where $\Omega_E$ is positive definite, that is, for $(x_1,\dots,x_6)\in\DD$. We seek the smallest positive value $C_{\mathrm{max}}$ of $C$ such that $\Omega(C)$ is not positive definite.  This is the same as the smallest value of $C$ such that 0 is an eigenvalue of the associated symmetric matrix $M_C$ to $\Omega(C)$, because the real eigenvalues of $M_C$ depend continuously on $C$ and are all positive for $C=0$.  In other words, $C_{\mathrm{max}}$ is the smallest positive solution of $\det(M_C)=0$.  Note that $C_{\mathrm{max}}$ certainly exists, because $\Omega_{\QQ}$ is not positive definite.

We find that
\begin{align}
    \det(M_C)=&\dfrac{x_1}{64}\left((4x_2x_6-x_3^2)x_5-x_6x_4^2\right)-\dfrac{C^2}{576}\left(x_5\left(x_6+36x_2-6x_3\right)-9x_4^2\right).
\end{align}
Thus, given $(x_1,\dots,x_6)\in\DD$, $\det(M_C)=0$ with $C\geq 0$ if and only if $C=C_{\mathrm{max}}$, where
\begin{align}\label{C_max}
    C_{\mathrm{max}}&=\sqrt{\dfrac{9x_1\left((4x_2x_6-x_3^2)x_5-x_6x_4^2\right)}{x_5\left(x_6+36x_2-6x_3\right)-9x_4^2}}.
\end{align}

For \eqref{C_max} to make sense the expression under the radical must be non-negative.  If the conditions of Proposition \ref{prop:energy-pos-def} hold, the following argument demonstrates this to indeed be the case. The numerator is positive by condition \emph{(ii)} and the fact that $x_1>0$. Condition \emph{(i)} alongside $x_2>0$ implies that $x_6>0$, thus we must also have $x_5>0$ by comparing this with condition \emph{(ii)}. Therefore, we have by condition \emph{(ii)} that
\begin{align}
    x_6x_5\left(x_6+36x_2-6x_3\right)-9x_6x_4^2>x_5\left(x_6^2+9x_3^2-6x_3\right)=x_5\left(3x_3-x_6\right)^2>0.
\end{align}
Thus the denominator is positive, since $x_6>0$. We have therefore proven the following.
\begin{theorem}\label{thm:energy-bound}
Let $(x_1,\dots,x_6)\in\DD\subset\R^6$ satisfy the inequalities in Proposition \ref{prop:energy-pos-def}. Then the energy \eqref{Gauged-Skyrme-energy-x-vars} is bounded below by the topological charge \eqref{top-charge} as
\begin{align}\label{main-top-bound}
    E\geq 24\pi^2\sqrt{\dfrac{x_1\left((4x_2x_6-x_3^2)x_5-x_6x_4^2\right)}{x_5\left(x_6+36x_2-6x_3\right)-9x_4^2}}|\QQ|.
\end{align}
\end{theorem}
\noindent The bound described by \eqref{main-top-bound} is a generalisation of, but does not improve upon, previously known topological bounds for gauged Skyrme models. For example, it is straightforward to check that with the couplings given by \eqref{Ultra-spherical-couplings-1}-\eqref{k_2}, this bound reduces to
\begin{align}\label{Ultraspherical-lower-bound}
    E_\alpha\geq 8\pi^2C(\alpha)|\QQ|,\quad C(\alpha)=\sqrt{\dfrac{9\kappa_0(2\kappa_1-\kappa_2^2)}{1+18\kappa_1-6\kappa_2}}
\end{align}
as was previously obtained in \cite{cork2018skyrmions}.  More traditional gauged Skyrme models such as those considered in \cite{PietteTchrakian2000static,ArthurTchrakian1996GaugedSky,BrihayeHartmannTchhrakian2001MonopolesGaugeSky,brihayeTchhrakian1998solitons} can be viewed as \eqref{Gauged-Skyrme-energy} with the choice of couplings $\chi_3=\chi_4=\chi_5=0$, which corresponds to $x_3=x_4=0$ and $x_5=x_6$ for \eqref{Gauged-Skyrme-energy-x-vars}. In this case, the conditions of Proposition \ref{prop:energy-pos-def} are simply that the remaining couplings $\chi_0= x_1,\chi_1= x_2,\chi_2= x_5=x_6$ are positive. In \cite{BrihayeHartmannTchhrakian2001MonopolesGaugeSky,brihayeTchhrakian1998solitons}, the couplings are defined by arbitrary parameters $\kappa_0,\kappa_1,\kappa_2$, which relate to our notation via $x_1=\frac{1}{2}\kappa_1^2$, $x_2=\frac{1}{8}\kappa_2^4$, $x_5=x_6=\kappa_0^4$. Simplifying \eqref{main-top-bound} with these parameters gives the bound
\begin{align}
    E\geq48\pi^2\dfrac{\kappa_1\kappa_2^2}{\sqrt{1+9\left(\frac{\kappa_2}{\kappa_0}\right)^4}}|\QQ|,
\end{align}
which corresponds (up to a rescaling of energy units) to the bound given in \cite{BrihayeHartmannTchhrakian2001MonopolesGaugeSky,brihayeTchhrakian1998solitons}.

Before moving on, we point out that the strategy adopted here could be improved in certain cases.  Instead of \eqref{Omega_E}, we could consider a quadratic form
\begin{align}
\Omega_E'(x,y,u,v)=x_1x^2+x_2y^2+\dfrac{x_6}{4}u^2+\dfrac{x_5}{4}v^2-\left(\dfrac{x_3}{2}u+\dfrac{x_4}{2}v\right)y+x_7 uv
\end{align}
for a non-zero real parameter $x_7$. This more general form is associated with the energy functional \eqref{Gauged-Skyrme-energy-x-vars} in the same way as \eqref{Omega_E}, because $\ip{F^{U,A}_+,F^{U,A}_-}_{L^2}=0$.  If $x_4\neq0$ this could lead to a better energy bound, but if $x_4=0$ and $x_5>0$ (the case of most interest to us) this does not lead to any improvement.  To see this, introduce new variables $\tilde{x}=x,\tilde{y}=y,\tilde{u}=u,\tilde{v}=v+2x_7u/x_5$.  Then $\Omega_E',\Omega_{\QQ}$ are equivalent to
\begin{align}
\widetilde{\Omega}_E'(\tilde{x},\tilde{u},\tilde{u},\tilde{v})&=x_1\tilde{x}^2+x_2\tilde{y}^2+\left(\dfrac{x_6}{4}-\dfrac{x_7^2}{x_5}\right)\tilde{u}^2+\dfrac{x_5}{4}\tilde{v}^2-\dfrac{x_3}{2}\tilde{u}\tilde{y} \\
\widetilde{\Omega}_{\QQ}(\tilde{x},\tilde{u},\tilde{u},\tilde{v})&=\tilde{x}\tilde{u}-\dfrac{1}{3}\tilde{x}\tilde{y}.
\end{align}
We see that $\widetilde{\Omega}_{\QQ}=\Omega_{\QQ}$ and $\widetilde{\Omega}_{E}'$ has the same form as $\Omega_E$ (with $x_4=0$), but $\tilde{\Omega}_E'$ is bounded from above by $\Omega_E$.  So any bound obtained for $\Omega_E'$ will be no better than one obtained for $\Omega_E$.

\subsection{Energy bounds restricted to the boundary of a domain}\label{sec:non-strict}
The bound \eqref{main-top-bound} applies for parameters $(x_1,\ldots,x_6)$ lying in an open domain $\DD\in\R^6$, defined by the conditions in Proposition \ref{prop:energy-pos-def}.  On the boundary of this domain the energy is still positive so it is of interest to look for energy bounds valid on the boundary of this domain.

Of particular interest in the next section is the boundary defined by the condition
\begin{align}\label{boundary-equation}
K:=4x_2x_5x_6-x_3^2x_5-x_4^2x_6=0.
\end{align}
Notice that the energy bound \eqref{main-top-bound} can be written in terms of $K$ as
\begin{align}
E\geq 24\pi^2\sqrt{\frac{x_1x_6K}{9K+x_5(x_6-3x_3)^2}}|\QQ|.
\end{align}
It helps to think about the limiting behaviour as $K\to0$.  Generically, as $K\to 0$, this lower bound tends to zero and we are unlikely to obtain a useful topological bound on the boundary component.  However, if both $K=0$ and $x_6=3x_3$ the numerator and denominator both vanish, and there is some hope of obtaining a non-trivial topological energy bound.

We therefore only consider the case
\begin{align}\label{restricted-condition}
    x_6=3x_3>0.
\end{align}
We assume moreover that $x_5>0$.  Substituting this back into \eqref{boundary-equation}-\eqref{restricted-condition} yields the condition
\begin{align}
    x_2=\dfrac{x_3}{12}+\dfrac{x_4^2}{4x_5}.
\end{align}
In this case, the energy \eqref{Gauged-Skyrme-energy-x-vars}, interpreted through the quadratic form \eqref{Omega_E}, takes the form
\begin{align}
    \Omega_E
    %&=x_1x^2+\left(\dfrac{x_3}{12}+\dfrac{x_4^2}{4x_5}\right)y^2+\dfrac{3x_3}{4}u^2+\dfrac{x_5}{4}v^2-\left(\dfrac{x_3}{2}u+\dfrac{x_4}{2}v\right)y\\
    &=x_1x^2+\dfrac{x_3}{12}\left(y-3u\right)^2+\dfrac{1}{4x_5}\left(x_4y-x_5v\right)^2.
\end{align}
Thus, for $x_3,x_5>0$, we may apply the following Bogomol'nyi-style completing the square argument: since
\begin{align}
    \Omega_E=\dfrac{x_3}{12}\left(2\sqrt{\dfrac{3x_1}{x_3}}x\pm\left(y-3u\right)\right)^2\pm\dfrac{x_3}{3}\sqrt{\dfrac{3x_1}{x_3}}x(3u-y)+\dfrac{1}{4x_5}\left(x_4y-x_5v\right)^2,\label{BOG-ineq}
\end{align}
comparing with \eqref{Omega_Q} we have the inequality $\Omega_E\geq\sqrt{3x_1x_3}|\Omega_{\QQ}|$.  Thus we have proved:
\begin{theorem}\label{thm:restricted-energy-bound}
Suppose that the $(x_1,\ldots,x_6)\in\R^6$ satisfy \eqref{boundary-equation} and \eqref{restricted-condition}, and that $x_1,x_5>0$.  Then the energy \eqref{Gauged-Skyrme-energy-x-vars} satisfies
\begin{align}\label{restricted-top-bound}
    E\geq 8\pi^2\sqrt{3x_1x_3}|\QQ|.
\end{align}
\end{theorem}
\subsection{Comparison with Yang--Mills bound}
Ultimately we are interested in versions of the gauged Skyrme model derived from Yang--Mills theory.  We now consider whether the bounds derived above apply in this situation, and compare them with the Yang--Mills topological energy bound.

For models derived from Yang--Mills theory the couplings are given by \eqref{YM-couplings-ints} and satisfy
\begin{align}\label{YM-x5-and-x6}
x_5=\beta-2x_3,\quad x_6=\beta.
\end{align}
The inequalities in Proposition \ref{prop:energy-pos-def} can then be written as
\begin{align}
    x_1,x_2&>0,\label{YM-couplings-cond1,2}\\
    4x_2\beta-x_3^2&>0,\label{YM-couplings-cond3}\\
    (4x_2\beta-x_3^2)(\beta-2x_3)-\beta x_4^2&>0.\label{YM-couplings-cond4}
\end{align}
These conditions are satisfied for any choice of function $\varphi_+$ satisfying
\begin{align}\label{phi-odd}
\varphi_+(-t)=1-\varphi_+(t).
\end{align}
To see this, note first that \eqref{YM-couplings-cond1,2} follows from \eqref{YM-couplings-ints} because $\varphi_+$ is not constant. Furthermore, the condition \eqref{YM-couplings-cond3} is equivalent to
\begin{align}\label{Cauchy--Schwarz}
    \left(\int_{-\frac{\beta}{2}}^{\frac{\beta}{2}}(1-\varphi_+)^2\varphi_+^2\;\d t\right)\left(\int_{-\frac{\beta}{2}}^{\frac{\beta}{2}}\;\d t\right)-\left(\int_{-\frac{\beta}{2}}^{\frac{\beta}{2}}(1-\varphi_+)\varphi_+\;\d t\right)^2>0,
\end{align}
which follows immediately from the Cauchy--Schwarz inequality, and is a strict inequality since $(1-\varphi_+)\varphi_+$ is not constant.  By \eqref{phi-odd} the integrand for $x_4$ in \eqref{YM-couplings-ints} is odd, so by symmetry $x_4=0$.  Then since \eqref{YM-couplings-cond3} is true, \eqref{YM-couplings-cond4} reduces to $\beta-2x_3>0$, which holds since
\begin{align}
\beta-2x_3=\int_{-\frac{\beta}{2}}^{\frac{\beta}{2}}(1-2\varphi_+)^2\;\d t>0.
\end{align}

If we drop the assumption \eqref{phi-odd} then the inequalities \eqref{YM-couplings-cond1,2} and \eqref{YM-couplings-cond3} are satisfied, but \eqref{YM-couplings-cond4} may not be. For example, if $\varphi_+$ is not monotonically increasing, such as
\begin{align}
\varphi_+(t)=\sin\left(\frac{3t\pi}{2}-\frac{\pi}{4}\right),
\end{align}
then
\begin{align}
    \left(4x_2-x_3^2\right)(1-2x_3)-x_4^2
    =-\frac{768 + 2048\pi + 1512\pi^2 + 243\pi^3}{162 \pi^3}
    %\approx-5.905
    <0.
\end{align}
One obstruction here is due to yielding $x_4\neq0$. However, in this case there is probably room for improvement in our analysis (see comments at the end of section \ref{subsec:gen-bounds}).

So, the bound \eqref{main-top-bound} applies to gauged Skyrme models which are derived from Yang--Mills and satisfy \eqref{phi-odd}.  The second bound \eqref{restricted-top-bound} applies to any Yang--Mills-derived gauged Skyrme model, provided the couplings satisfy the boundary-defining conditions \eqref{boundary-equation} and \eqref{restricted-condition}.  To see this, note that \eqref{restricted-condition} and \eqref{YM-x5-and-x6} imply that $x_3=x_5=\beta/3$.  So the coupling constants satisfy all of the hypotheses for the bound \eqref{restricted-top-bound}, which takes the form
\begin{align}\label{restricted-top-bound2}
E\geq 8\pi^2\sqrt{x_1\beta}|\QQ|.
\end{align}

A gauged Skyrme model that is derived from Yang--Mills theory satisfies a third, more obvious, topological energy bound, derived from the bound on the Yang--Mills energy.  It is well-known that the Yang--Mills energy satisfies the inequality
\begin{align}\label{YM-top-bound}
\int \tr(F^{\tilde{A}}\wedge\star(F^{\tilde{A}})^\dagger) \geq \left|\int \tr(F^{\tilde{A}}\wedge F^{\tilde{A}})\right|.
\end{align}
Since the Yang--Mills topological charge $\frac{1}{8\pi^2}\int\tr(F^{\tilde{A}}\wedge F^{\tilde{A}})$ coincides with \eqref{Gauged-Skyrme-charge}, the gauged Skyrme model topological charge \cite{cork2018skyrmions}, this leads to a bound $E\geq 8\pi^2|\QQ|$ in any gauged Skyrme model derived from Yang--Mills theory.  Our second bound \eqref{restricted-top-bound2} is stronger than the Yang--Mills bound, because
\begin{align}
    \sqrt{x_1\beta}=\sqrt{\int_{-\frac{\beta}{2}}^{\frac{\beta}{2}}{\varphi_+'(t)}^2\;\d t\int_{-\frac{\beta}{2}}^{\frac{\beta}{2}}\;\d t}\geq\int_{0}^1\varphi_+\;\d\varphi_+=1
\end{align}
by the Cauchy--Schwarz inequality.  We do not know in general whether the first bound \eqref{main-top-bound} is also stronger than the Yang--Mills bound, but it does seem to be stronger in certain cases.  For example, in \cite{cork2018skyrmions} it was verified numerically that for the ultraspherical family \eqref{hypergeometric-fct}, the bound \eqref{Ultraspherical-lower-bound} is always stronger than the Yang--Mills bound, and in fact one can straightforwardly see it agrees with the Yang--Mills bound when $\alpha=0$.

\section{Optimising the BMI in gauged Skyrme models}\label{sec:opt-BMI}
In this section we aim to optimise the BMI \eqref{BMI} for gauged skyrmions, for each of the bounds considered in the previous section. As it is dimensionless, the BMI allows a meaningful comparison between models with different choices of coupling constants. We approximate minimisers by using caloron-generated gauged Skyrme fields, and compare with numerical minimisers of the energy functional. In particular, we consider the problem of optimising the BMI from two different angles. The first, and most obvious approach that we use is to fix a set of parameters $(x_1,\dots,x_6)\in \DD\subset\R^6$, which we do via \eqref{Ultra-spherical-couplings-1} derived from the ultraspherical kink \eqref{hypergeometric-fct}, and then minimise over all field configurations $(U,A)$ satisfying certain boundary conditions, which we describe in the following sections. The second approach is to use a fixed caloron configuration $(U,A)$, and to minimise over all coupling constants $(x_1,\dots,x_6)\in\ol{\DD}\subset\R^6$. From this, we obtain an optimal set of couplings which lie on the boundary $\bdy\DD$, from which we find numerical minimisers and compare to the caloron approximation.
\subsection{Calorons as gauged skyrmions}\label{subsec:skyrme-from-cal}
Calorons are anti-self-dual $\SU(2)$ gauge fields $\wt{A}$ on $S^1\times\R^3$.  They saturate the Yang--Mills topological bound \eqref{YM-top-bound}, so are minima of the Yang--Mills energy.  Therefore they give good candidates for approximating minimisers of a gauged Skyrme energy derived from Yang--Mills theory. 

\subsubsection{KvBLL calorons}
We consider the energy of gauged skyrmions obtained from a particular family of calorons, the Kraan--van Baal--Lee--Lu (KvBLL) calorons \cite{KraanVanBaal1998,LeeLu1998}, and compare with numerically-obtained minimisers of the gauged Skyrme energy.  The KvBLL calorons give gauged skyrmions with $\QQ=-1$.  Other calorons and monopoles with $|\QQ|\leq 1$, and the corresponding numerical gauged skyrmions were studied earlier in \cite{cork2018skyrmions}.

The main reason to study the KvBLL calorons is that they satisfy a boundary condition
\begin{align}\label{boundary-condition}
U \to -{\rm i}\sigma^3\;\text{ as }\;r\to\infty,
\end{align}
which breaks the gauge symmetry to the physically realistic $\U(1)$.  As a result, they have only axial symmetry and not spherical symmetry.  As we will see, their energy is closer to the topological bound than that of the spherically-symmetric $|\QQ|=1$ calorons studied in \cite{cork2018skyrmions}. In general, calorons satisfy a boundary condition of the form\footnote{For a full description of the caloron boundary conditions, see \cite{Nyethesis}.}
\begin{align}\label{caloron-bdy}
    \wt{A}_t=\left(\mu-\dfrac{m}{2r}\right){\rm i}\sigma^3+O(r^{-2}),\quad \wt{A}_r=O(r^{-3}),
\end{align}
where $t$ is a coordinate on $S^1$ and $r$ is the radial coordinate in $\R^3$. Here $0\leq\mu<\pi/\beta$ is called the \textit{holonomy parameter} and $m\in\Z$ is a topological charge called the \textit{magnetic charge} of the caloron. Calorons possess a second topological charge $k\in\Z$ called the \textit{instanton number}, which roughly speaking measures the winding of the caloron on the interior of $S^1\times\R^3$. In the cases where $m=0$, the holonomy is constant at infinity, and therefore has a well-defined degree, and this is precisely $-k$, (minus) the instanton number.

Calorons are often thought of in terms of two \emph{constituent monopoles} \cite{kraanVanBaal1998monopoleconsts} of charges $(m_1,m_2)$ and masses $(\nu_1,\nu_2)$. These quantities are determined by the magnetic charge, instanton number, holonomy parameter, and the period:
\begin{align}
    (m_1,m_2)=(m+k,k),\quad(\nu_1,\nu_2)=(2\mu,\mu_0-2\mu), \quad \mu_0:=\frac{2\pi}{\beta}.
\end{align}
The total charge \eqref{YM-charge} of such a caloron may be computed through these numbers as \cite{Nyethesis}
\begin{align}\label{caloron-charge}
    Q=-(m_1\nu_1+m_2\nu_2).
\end{align}

Calorons for which $\mu=0$ are said to have \textit{trivial holonomy}, whereas calorons with $0<\mu<\mu_0/2$ are said to have \textit{non-trivial (asymptotic) holonomy}. The first calorons with non-trivial holonomy and positive instanton number were constructed independently by Kraan and van Baal \cite{KraanVanBaal1998}, and Lee and Lu \cite{LeeLu1998}.  These have magnetic charge $m=0$ and instanton number $k=1$, and hence from \eqref{caloron-charge}, the corresponding gauged skyrmions have topological charge $\QQ=-1$.

The fundamental idea behind the construction of the KvBLL calorons is that of the \textit{Nahm transform for calorons} \cite{Nyethesis,KraanVanBaal1998Tdual,CharbonneauHurtubise2007nahm.tfm}. A more detailed review of KvBLL calorons and the Nahm transform is given in appendix \ref{appendix:KvBLL}, which is where our full conventions are clarified.  The full solution depends on eight moduli space parameters, and the holonomy parameter $\mu$. Of the moduli space parameters, four of these correspond to translation and, as the solution is axially-symmetric, two correspond to rotations of $\R^3$. There is also an overall phase between the constituent monopoles corresponding to a global $\U(1)\subset\SU(2)$ gauge transformation. The remaining parameter is a scale parameter $\lambda$.  When $\lambda$ is small the caloron resembles a charge $-1$ instanton located at a point in $S^1\times\R^3$.  When $\lambda$ is large the caloron resembles a pair of monopoles in $\R^3$ separated by a distance $\lambda^2/2$.  Their masses are $2\mu$ and $\mu_0-2\mu$, and when $\mu=\mu_0/4$ the masses are equal and the solution acquires additional discrete symmetries that swap the two constituent monopoles \cite{cork2018symmrot}.

The resulting family of skyrmions will be called \emph{KvBLL skyrmions}, and we will calculate their energy below. We remark that these will not be skyrmions in the strict sense of satisfying the field equations, but only approximations to the true minimisers. For convenience we fix $\beta=1$ (and hence $\mu_0=2\pi$).  No generality is lost here, because results for $\beta\neq1$ are easily obtained by a simple rescaling of parameters in the gauged Skyrme energy.  We will also fix $\mu=\mu_0/4=\pi/2$, as these calorons have additional discrete symmetries, and this choice simplifies a lot of later calculations.  We will fix the location of the caloron at the origin in $S^1\times\R^3=(\R/\Z)\times\R^3$, and rotate it so that it is axially symmetric about the $X_3$-axis in $\R^3$. We also fix the phase between the constituents to $-\frac{\pi}{2}$. Although transformations of $\R^3$ can be used to fix the location and orientation in $\R^3$, the way that we extract a skyrmion from a caloron breaks the translation symmetry in $S^1$, so fixing the location in $S^1$ entails some loss of generality.  Our choice to locate at $t=0$ seems natural in that the caloron is equidistant from the two hypersurfaces $t=\pm\beta/2$ between which we calculate holonomy.  Our choices then leave a single parameter $\lambda>0$ that can be varied.

\subsubsection{The energy of KvBLL skyrmions}
To calculate the energy of the KvBLL skyrmions we need to know the components $A_j$ of the gauge field and the Skyrme field $U$.  The former are easily obtained by evaluating the components of the caloron gauge field at time $t=-\beta/2$, which is found using the Nahm transform (see appendix \ref{appendix:KvBLL}).  The latter corresponds to holonomy of the caloron about circles parameterised by $t$, and is difficult to calculate because the calorons that we work with have a gauge singularity at $(t,\vec{X})=(0,\vec{0})$.  We circumvent this difficulty using a geometric interpretation of the Nahm transform, which we briefly describe here.

The Nahm transform generates an orthonormal frame $\{e^1,e^2\}$ for a rank $2$ subbundle $V$ of an infinite-rank trivial hermitian vector bundle over $S^1\times\R^3$. The infinite-rank bundle carries a flat connection, and the caloron is the induced connection on the subbundle.  This is the connection given by the $2\times2$ matrix of $1$-forms with components
\begin{align}
    \wt{A}^{ab}=\ip{e^a,\d e^b}.
\end{align}
For details, see appendix \ref{appendix:KvBLL}.  The Skyrme field $U$ is the holonomy of this connection, and we approximate it using the formula
\begin{align}\label{approximate-U}
    U(\vec{X})\approx \Omega_{\vec{X}}(t_k,t_{k-1})\Omega_{\vec{X}}(t_{k-1},t_{k-2})\cdots\Omega_{\vec{X}}(t_{1},t_{0}).
\end{align}
Here $-\beta/2=t_0<t_1<\ldots<t_k=\beta/2$ are evenly spaced points on the interval $[-\beta/2,\beta/2]$ and $\Omega_{\vec{X}}(t_{i},t_{i-1})$ are $2\times 2$ matrices with components
\begin{align}\label{approximate-parrallel-transport}
    \Omega_{\vec{X}}^{ab}(t_i,t_{i-1})=\ip{e^a(t_i,\vec{X}),e^b(t_{i-1},\vec{X})}.
\end{align}
These matrices describe orthogonal projection from the span of $\{e^1,e^2\}$ at $(t_{i-1},\vec{X})$ to the span at $(t_{i},\vec{X})$.  When $|t_i-t_{i-1}|$ is small this approximates the holonomy of the connection $\wt{A}$ between these points, and so \eqref{approximate-U} approximates the holonomy from $t=-\beta/2$ to $t=\beta/2$.

An advantage of using the formula \eqref{approximate-U}, rather than solving the differential equation \eqref{holonomy-cal-sky} directly, is that \eqref{approximate-U} respects gauge covariance.  If we make a gauge transformation $e^a(t,\vec{X})\mapsto e^b(t,\vec{X})g^{ba}(t,\vec{X})$ then $\Omega_{\vec{X}}$ and $U$ transform as
\[ \Omega_{\vec{X}}(t_i,t_{i-1})\mapsto g(t_i,\vec{X})^{-1}\Omega_{\vec{X}}(t_i,t_{i-1})g(t_{i-1},\vec{X}),\quad U(\vec{X}) \mapsto g(t_k,\vec{X})^{-1}U(\vec{X})g(t_0,\vec{X}). \]
In particular, $U$ is insensitive to choices of gauge at points $t_i$ between $\pm\beta/2$, so the gauge singularity at $t=0$ is of no consequence.

One caveat is that the matrix \eqref{approximate-U} is only approximately unitary.  To make it exactly unitary we multiply it on the right by $(U^\dagger U)^{-\frac{1}{2}}$ (where the square root is the canonical choice from the spectral decomposition).
\begin{table}[ht]
    \centering
    \begin{tabular}{c|cccc|c}
    $\lambda$&$\epsilon_0$&$\epsilon_1$&$\epsilon_2$&$\epsilon_3$&$|\QQ+1|$\\\hline
     0.1&0.004&4.046&21402.6&0.037&$3.4\times10^{-5}$\\
    0.2&0.072&8.090&10602.2&0.563&$4.2\times10^{-5}$\\
    0.3&0.349 &12.072&6918.37&2.526&$4.8\times10^{-5}$\\
    0.4&1.024&15.898 &4990.21&6.651&$5.4\times10^{-5}$\\
    0.5&2.276&19.461&3755.35&12.922&$6.0\times10^{-5}$\\
    0.6&4.213&22.670&2870.73&20.562&$6.6\times10^{-5}$\\
    0.7&6.849&25.457&2198.79&28.385&$7.2\times10^{-5}$\\
    0.8&10.089&27.791&1676.51&35.228&$7.8\times10^{-5}$\\
    0.9&13.759&29.672&1270.54&40.271&$8.6\times10^{-5}$\\
    1.0&17.635&31.127&959.301&43.163&$9.4\times10^{-5}$\\
    1.1&21.487&32.204&725.497&44.000&$1.1\times10^{-4}$\\
    1.2&25.118&32.968&553.905&43.152&$1.3\times10^{-4}$\\
    1.3&28.383&33.487&430.629&41.131&$1.6\times10^{-4}$\\
    1.4&31.208&33.831&343.509&38.448&$1.8\times10^{-4}$\\
    1.5&33.599&34.058&282.484&35.511&$2.2\times10^{-4}$\\
    1.6&35.653&34.218&239.719&32.568&$2.6\times10^{-4}$\\
    1.7&37.609&34.350&209.440&29.653&$3.0\times10^{-4}$\\
    1.8&39.999&34.480&187.592&26.458&$3.7\times10^{-4}$\\
    1.9&44.082&34.627&171.439&21.939&$4.6\times10^{-4}$\\
    2.0&53.114&34.813&159.200&13.043&$6.0\times10^{-4}$
    \end{tabular}
    \caption{The numerically evaluated coefficients of the energy function $E(x)$ for the KvBLL skyrmion with scale $0.1\leq\lambda\leq 2$. For this caloron $\epsilon_4=0$, and so it is omitted. The accuracy is measured by the quoted value of $|\QQ+1|$.}
    \label{tab:KvBLL-cal-energy-terms}
\end{table}

After calculating the KvBLL Skyrme fields, we calculated the various components of the Skyrme energy \eqref{Gauged-Skyrme-energy-x-vars}. In particular, so as to compare with models derived from Yang--Mills theory with couplings \eqref{YM-couplings-ints}, we consider the combinations
\begin{align}\label{energy-values}
\begin{array}{c}
\begin{array}{ccc}
    \epsilon_0=\norm{F^A}^2,&\epsilon_1=\norm{L^A}^2,&\epsilon_2=\norm{L^A\wedge L^A}^2,\end{array}\\
    \epsilon_3=\ip{F^A,U^{-1}F^AU-F^A-\frac{1}{2}(R^A\wedge R^A+L^A\wedge L^A)},\\
    \epsilon_4=\frac{1}{2}\ip{U^{-1}F^AU-F^A,L^A\wedge L^A}.
    \end{array}
\end{align}
The energy \eqref{Gauged-Skyrme-energy-x-vars}, with $x_5$ and $x_6$ fixed by \eqref{YM-x5-and-x6} with $\beta=1$, is then given by
\begin{align}\label{E(x)}
    E(x)=\epsilon_0+\epsilon_1x_1+\epsilon_2x_2+\epsilon_3x_3+\epsilon_4x_4.
\end{align}
In practice these integrals were evaluated numerically on a rectangular grid in $\R^3$.  In order to accommodate the dependence of the skyrmion size on $\lambda$, the dimensions of this box were taken to be $a\lambda\times a\lambda\times(a\lambda+\lambda^2/2)$ for a large positive constant $a$.  The values of $\epsilon_\mu$ obtained are given in table \ref{tab:KvBLL-cal-energy-terms}, along with the deviation $|\QQ+1|$ of the numerically evaluated topological charge from the true charge as a measure of accuracy.  Values of $\epsilon_\mu$ for intermediate values of $\lambda$ were calculated from these using polynomial interpolation.
\subsection{Numerical minimisation}
In the following sections we shall demonstrate that the KvBLL skyrmions provide good approximations to true energy minimisers. To measure how good this approximation is, we will find numerical minimisers of the energy and compare with the constructed KvBLL skyrmions, and in doing so we will demonstrate that both exhibit low BMIs when compared to other Skyrme models.

Specifically, we will evolve a discretised approximation of the system using a gradient descent method. The number of terms in the energy functional makes calculating the gradient a numerically demanding task compared with other models, limiting the density of lattice points. Due to this and the changing length scales of various fields, it will prove useful to perform the change of coordinates $X_i = \tan(Y_i)$ for the numerical solutions, where $X_i$ are the standard cartesian coordinates on $\mathbb{R}^3$ and $Y_i \in [-\pi/2, \pi/2]$.

The model was simulated on a regular three-dimensional grid of $N_1 \times N_2 \times N_3$ lattice sites with spacing $h_i = \pi/(N_i -1)$, such that $Y_i =n_i h_i-\pi/2$, $n_i \in [0,N_i-1]\cap\Z$.  The plots in this section were simulated with values $N_1 = N_2 = N_3 = 101$, with fixed boundary conditions.  We approximate the 1st and 2nd order spatial derivatives with respect to $Y_i$ using central 4th order finite difference operators, yielding a discrete approximation $E_{\rm dis}$ to the functional $E[\psi]$ in \eqref{expanded-energy}, where $\psi = (\phi, A)$. We therefore aim to minimise the function $E_{\rm dis} : \CC\lto \mathbb{R}$, where the discretised configuration space is the manifold $\CC = (S^3\times\su(2)^3)^{N_1 N_2 N_3} \subset (\mathbb{R}^4 \times \mathbb{R}^9)^{N_1 N_2 N_3}$, and we represent $\phi$ as a 4-vector subject to the constraint $\phi\cdot\phi = 1$.  To find local minima of $E_{\rm dis}$ we use an arrested Newton flow algorithm, solving for the motion of a particle in $\CC$ with potential $E_{\rm dis}$,
\begin{align}
\ddot{\psi} = - \mbox{grad} E_{\rm dis}(\psi),
\end{align}
starting at an initial configuration $\psi(0)$ and $\dot{\psi} = 0$.  Evolving this algorithm naturally flows the configuration towards a local minima. At each time step $t \mapsto t + \delta t$, we check to see if the direction of the force on the particle has reversed. If $\ddot{\psi}\cdot \dot\psi < 0$, we set $\dot\psi = 0$ and continue relaxing the configuration. The flow is terminated once the discrete approximate is sufficiently close to a local minima, namely when every component of $\mbox{grad} E_{\rm dis}(\psi)$ is zero within a given tolerance.

We already have an ideal initial condition in the form of the KvBLL skyrmion. Therefore we fix the boundary condition to be that of the KvBLL skyrmion without loss of generality and relax the constructed configuration.

\subsection{Optimising within the ultraspherical family}

Recall that the ultraspherical family is a one-parameter family of gauged Skyrme models, with couplings given in \eqref{Ultra-spherical-couplings-1} in terms of a parameter $\alpha>-\frac12$.  For a range of values of $\alpha$, we determined the value of the scale parameter $\lambda$ that minimises the BMI \eqref{BMI} of a KvBLL skyrmion.  The results are plotted in subfigure \ref{subfig:BMI-Ultraspherical-couplings-lambdamin}.
\begin{figure}[ht]
\centering
\begin{subfigure}{.5\textwidth}
\centering
\includegraphics[width=\linewidth]{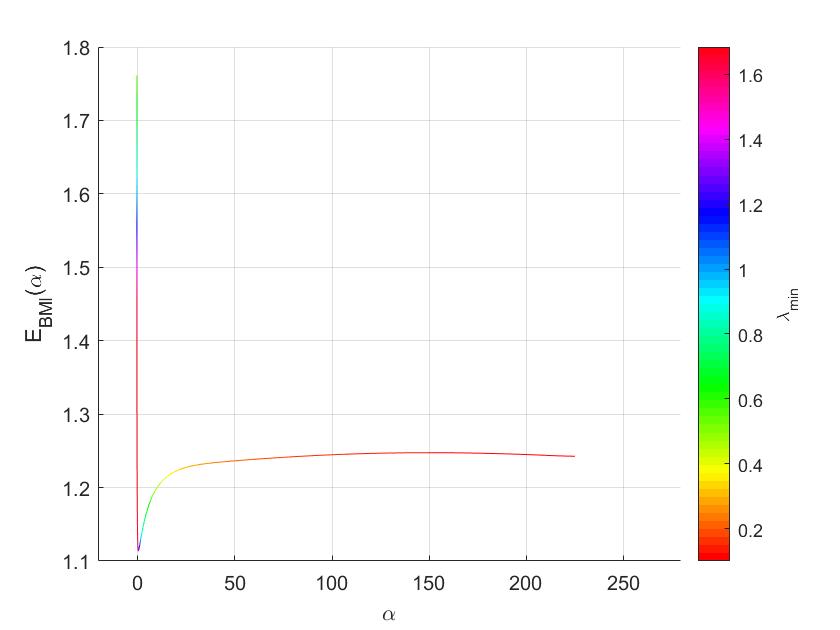}
\caption{}
\label{subfig:BMI-Ultraspherical-couplings-lambdamin}
\end{subfigure}%
\begin{subfigure}{.5\textwidth}
\centering
\includegraphics[width=\linewidth]{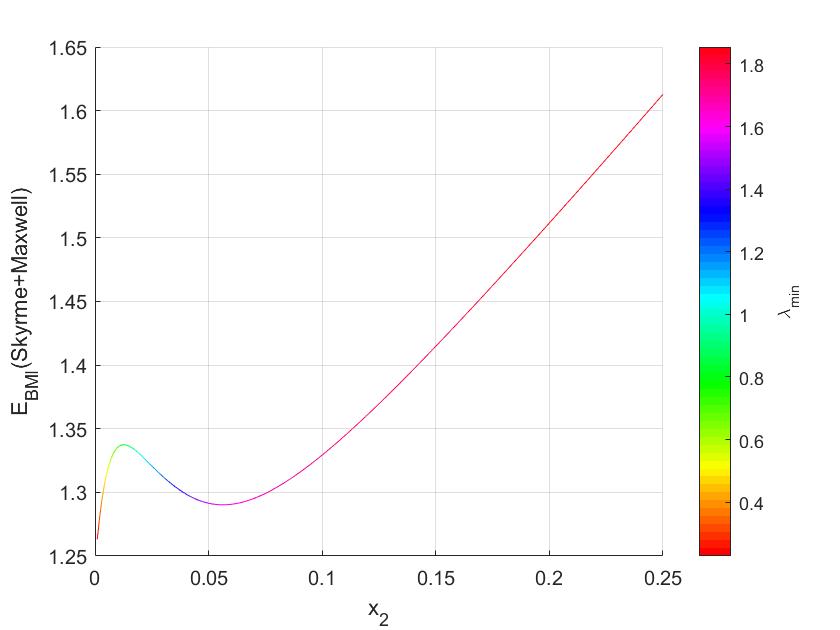}
\caption{}
\label{subfig:KvBLL-in-Maxwell}
\end{subfigure}
\caption{The BMI of the optimal KvBLL skyrmion in (a) the gauged model from the ultraspherical kink, for $\alpha>-\frac{1}{2}$, compared with (b) the traditional gauged Skyrme model $E=E_{\mathrm{Skyrme}}+E_{\mathrm{Maxwell}}$ as a function of the Skyrme coupling $x_2>0$.}
\label{fig.BMI-compare-different-models}
\end{figure}

The BMIs that we obtain compare well with other versions of the gauged Skyrme model.  When $\alpha$ is positive and close to zero (so that the Skyrme field and gauge field are strongly coupled) the BMI is close to $1.11$, so quite low.  For comparison, in the $\U(1)$-gauged model studied in \cite{PietteTchrakian2000static} BMIs at strong coupling were around $1.82$, so much higher.  At weak coupling ($\alpha\to\infty$) our BMIs are close to the value $1.23$ obtained in the ordinary ungauged Skyrme model. We attribute the low BMI of the KvBLL skyrmion to both the choice of coupling constants and the choice of boundary condition.
\subsubsection{Comparison with traditional model}
First, to show that the coupling constants are well-chosen for a low BMI, we have calculated the optimal BMI of the KvBLL skyrmions in the traditional gauged Skyrme model, which is given by the energy \eqref{E(x)} with $x_3=x_4=0$. We fixed $x_1=1$, which together with the already fixed $x_5=x_6=\beta=1$, amounts to a rescaling of energy and length units within this restricted class of models. We are left with one coupling parameter: the Skyrme coupling $x_2>0$. Up to this rescaling of units, varying $x_2$ is equivalent to varying the gauge coupling, with $x_2\to0$ and $x_2\to\infty$ measuring the weak and strong coupling limits respectively. The BMIs for the optimal KvBLL skyrmions are plotted in subfigure \ref{subfig:KvBLL-in-Maxwell}. As can be seen by comparing to subfigure \ref{subfig:BMI-Ultraspherical-couplings-lambdamin}, the model determined by the ultraspherical family exhibits significantly smaller BMIs at strong coupling, whereas the models are comparable at weak coupling as should be expected.
\subsubsection{Comparison with other boundary conditions}
To show that the boundary condition \eqref{boundary-condition} is well-chosen for a low BMI, we compare the BMI of the KvBLL calorons with two other types of caloron which were studied earlier in \cite{cork2018skyrmions}, namely the Harrington--Shepard calorons \cite{HarringtonShepard1978periodic} and the Prasad--Sommerfeld monopole \cite{PrasadSommerfield1975}.  Both have $|\QQ|=1$, but unlike the the KvBLL caloron they have spherical symmetry and satisfy the boundary condition $U\to\mathbb{1}$ as $r\to\infty$. The gauge field $A$ has different asymptotic behaviour in each case.
\begin{figure}[ht]
\centering
\begin{subfigure}{.5\textwidth}
\centering
\includegraphics[width=\linewidth]{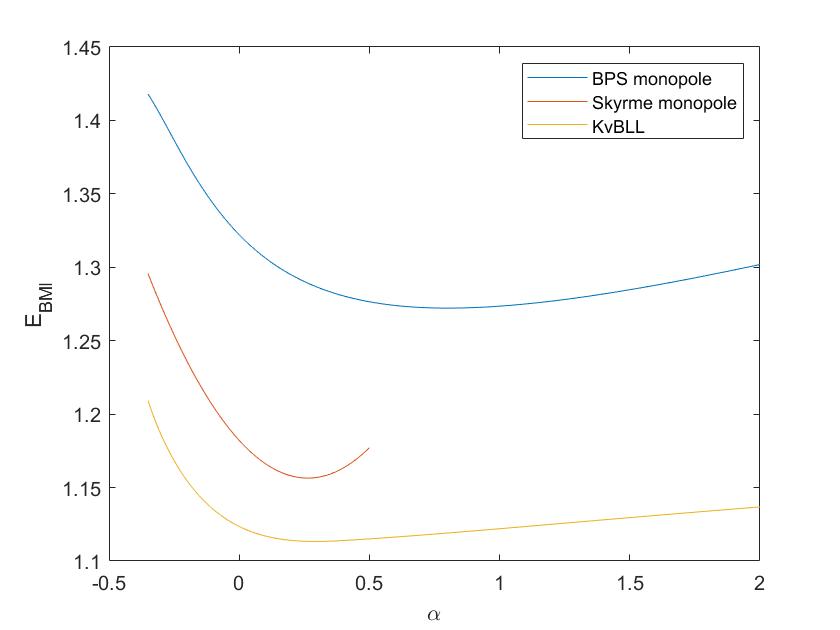}
\caption{}
\label{subfig.EBMI-compare-monopoles}
\end{subfigure}%
\begin{subfigure}{.5\textwidth}
\centering
\includegraphics[width=\linewidth]{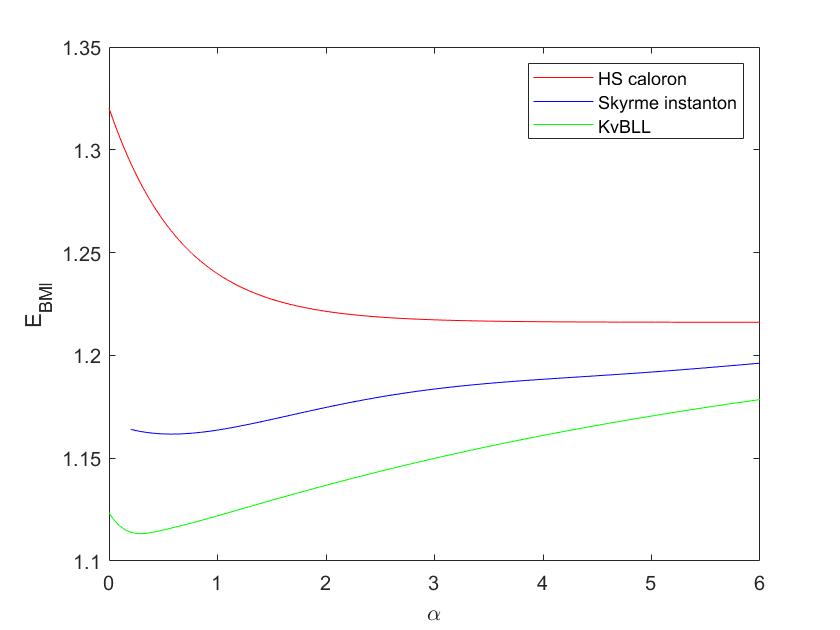}
\caption{}
\label{subfig.EBMI-compare-instantons}
\end{subfigure}
\caption{The BMI of the KvBLL skyrmions compared to other charge $|\QQ|=1$ gauged Skyrme fields -- minimisers and caloron approximations.}
\label{fig.BMI-compare-other-skyrmions}
\end{figure}
The BMIs are plotted in figure \ref{fig.BMI-compare-other-skyrmions}, along with numerical minimisers obtained in \cite{cork2018skyrmions} -- referred to as \textit{Skyrme-monopoles} and \textit{Skyrme-instantons} respectively. It is important to note in this comparison that in \cite{cork2018skyrmions} no numerical minimisers were found for certain values of $\alpha$ in each case, hence the obvious truncation in the plots. The main observation is that the KvBLL skyrmion has a lower BMI than all of these spherically-symmetric configurations, including the numerical solutions. This is surprising, because the KvBLL solution has less symmetry than these other configurations (axial rather than spherical).  In almost all known models that support topological solitons, the most symmetric solution of the field equations with topological charge 1 minimises the energy, so the gauged Skyrme model is very unusual.
\subsubsection{Numerical minimisers}
We have shown that by optimising over the moduli space we can find a KvBLL skyrmion that has a comparatively low BMI.  However,  as the configurations constructed from KvBLL calorons are not true solutions of the equations of motion, we will compare these with numerical solutions. Using the KvBLL skyrmions from the previous section as an initial condition and relaxing via Newton flow, we find a local minimiser of the gauged Skyrme model \eqref{Gauged-Skyrme-energy-x-vars}.  The local minimiser for $\alpha = 1$ is plotted in figure \ref{fig.skyrmePlot} and has a BMI of $E_{\rm BMI} = 1.104 $.

The colouring in the plot indicates the pion field direction; writing $U=-{\rm i}\sigma^3\exp({\rm i}\pi_j\sigma^j)$, the skyrmion is white/black when $\pi_3 = \pm 1$, and red, green, or blue when $\pi_1 + i \pi_2 = 1$, $\exp(2\pi i/3)$, or $\exp(4\pi i/3)$ respectively. We remark that, unlike the standard Skyrme model, where the choice of colouring only fixes a global iso-orientation, here the colouring scheme is not gauge-invariant. However, the fact that every colour appears once should reassure the reader that the topological charge is indeed 1. Finally, note that the degree 1 minimiser has only axial symmetry, as expected from the caloron model.

\begin{figure}[ht]
\centering
\includegraphics[width=0.5\linewidth]{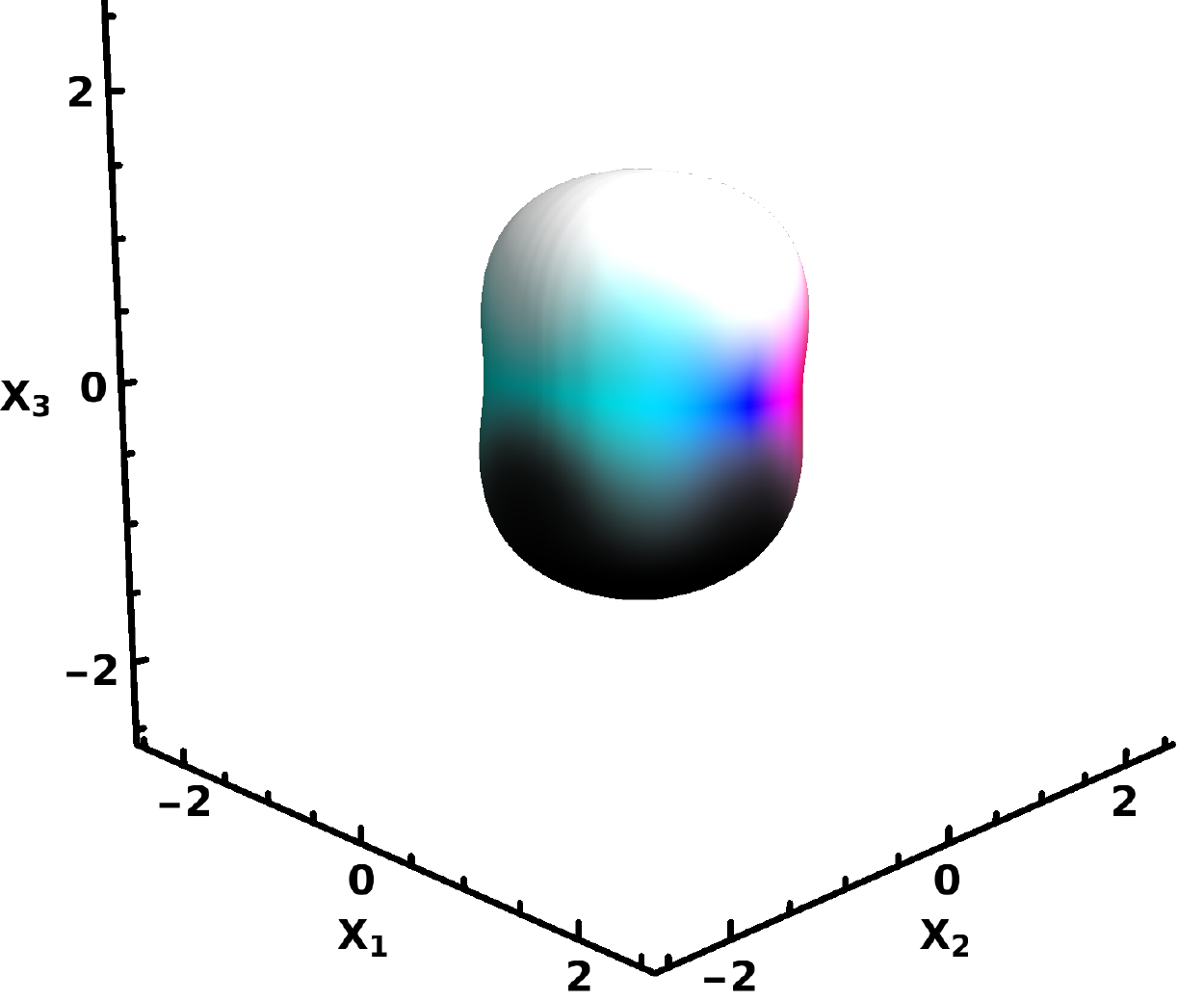}
\caption{Isosurface plot of the energy density for a local minimiser of the gauged Skyrme model. The model used ultraspherical couplings with $\alpha=1$. The plot is coloured by the direction of the fields $(\phi_1,\phi_2,\phi_3)$ and the isosurface is at the value $\EE = 0.5$. }
\label{fig.skyrmePlot}
\end{figure}

The BMI for $\alpha = 1$ is not much smaller than the approximation given by the caloron.  To see this more generally, we can compare the KvBLL skyrmions with the minimisers of the model in figure \ref{fig.energyComparePlot}. In this figure we have plotted the BMIs of the KvBLL skyrmions and the corresponding minimisers. It can be seen that the BMIs of the true solutions are not significantly below those of the KvBLL skyrmions. 

\begin{figure}[ht]
\centering
\includegraphics[width=0.6\linewidth]{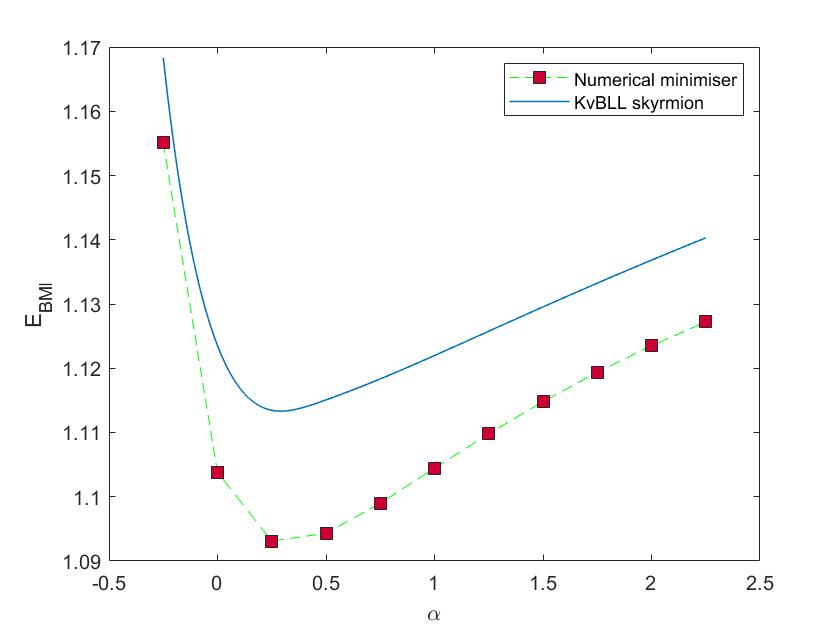}
\caption{The BMI of the true numerical minimisers (red dots) compared to the BMI of the KvBLL skyrmion (blue line) for varying $\alpha>-\frac{1}{2}$ in the ultraspherical model.}
\label{fig.energyComparePlot}
\end{figure}

\subsection{Optimising over all coupling constants}
So far we have seen that the ultraspherical family of gauged Skyrme models exhibits low BMIs at strong coupling, and the KvBLL calorons provide good approximations to numerical minima. We would like to know how good the BMI can get in more generality. To do this, we shall now flip the problem on its head; instead of fixing the couplings, and optimising fields to minimise the energy, we shall fix the caloron, and optimise the coupling constants to minimise the energy. Again we restrict to cases where the couplings are determined from Yang--Mills theory, which amounts to fixing $x_5$ and $x_6$ as \eqref{YM-x5-and-x6} so that the energy is given by \eqref{E(x)}.

For fixed $\epsilon_\mu$, the BMI \eqref{BMI} for either of the bounds \eqref{main-top-bound} and \eqref{restricted-top-bound} may be thought of as a function $E_{\mathrm{BMI}}:\ol{\DD}\lto[1,\infty)$, where here $\DD\subset\R^4$ is defined by the couplings $(x_1,x_2,x_3,x_4)$ satisfying the inequalities \eqref{YM-couplings-cond1,2}-\eqref{YM-couplings-cond4} (with $\beta=1$), where \eqref{YM-couplings-cond4} is replaced by an equality in the case of the bound \eqref{restricted-top-bound}, representing the bound restricted to $\bdy\DD$. We wish to minimise $E_{\mathrm{BMI}}$ for the KvBLL skyrmion, which is equivalent to minimising
\begin{align}\label{BMI-squared}
\left(24\pi^2E_{\mathrm{BMI}}\right)^2 = \left[ \epsilon_0 + \sum_{i=1}^3x_i\epsilon_i\right]^2 
\frac{1}{x_1}\left[9 + \frac{(1-2x_3)(1-3x_3)^2}{(1-2x_3)(4x_2-x_3^2)- x_4^2}\right]
\end{align}
with respect to $x_1,\ldots,x_4$, subject to the constraints \eqref{YM-couplings-cond1,2}-\eqref{YM-couplings-cond4}. Note that we have set $\epsilon_4=0$ as this is the case for the KvBLL skyrmion. There is an exact solution to this problem which determines \textit{optimal couplings} in terms of $\epsilon_0,\dots,\epsilon_3$.
\begin{proposition}\label{prop:optimal-couplings}
Suppose that $(x_1,\dots, x_4)\in\R^4$ satisfy $x_1,x_2>0$,
\begin{align}
    4x_2-x_3^2\geq0,\label{YM-constraint-3-relaxed}
\end{align}
and\footnote{Note that these are conditions \eqref{YM-couplings-cond3}-\eqref{YM-couplings-cond4} relaxed to allow for equality.}
\begin{align}
    (4x_2-x_3^2)(1-2x_3)-x_4^2\geq 0.\label{YM-constraint-4-relaxed}
\end{align}
Let $\epsilon_0,\dots,\epsilon_3>0$ be positive constants such that $\epsilon_0\epsilon_2>\epsilon_3^2$. Then \eqref{BMI-squared} satisfies
\begin{align}\label{energy-bound-optimal}
\left(24\pi^2E_{\mathrm{BMI}}\right)^2 \geq \epsilon_1(\epsilon_2+12\epsilon_3+12\epsilon_0),
\end{align}
with equality if and only if
\begin{align}\label{optimal-couplings}
x_1=\frac{1}{\epsilon_1}\left[\epsilon_0+\sum_{i=2}^3x_i\epsilon_i\right],\quad x_2=\frac{1}{36},\quad x_3=\frac13,\quad x_4=0.
\end{align}
\end{proposition}
Before we prove Proposition \ref{prop:optimal-couplings}, note that the couplings \eqref{optimal-couplings} lie on the boundary $\bdy\DD$ determined by the condition \eqref{boundary-equation}. Furthermore, the equality in \eqref{energy-bound-optimal} is consistent with the bound \eqref{restricted-top-bound} found previously. The only thing that is not obvious is that the energy values \eqref{energy-values} satisfy the constraint $\epsilon_0\epsilon_2-\epsilon_3^2>0$, however this can be verified from the data in table \ref{tab:KvBLL-cal-energy-terms}, and via the plot given in figure \ref{fig:epsilon-constraint}, which includes the values obtained via interpolation. We remark that this quantity appears only to tend to $0$ as $\lambda\to0$.\\
\begin{figure}[ht]
    \centering
    \includegraphics[width=0.6\linewidth]{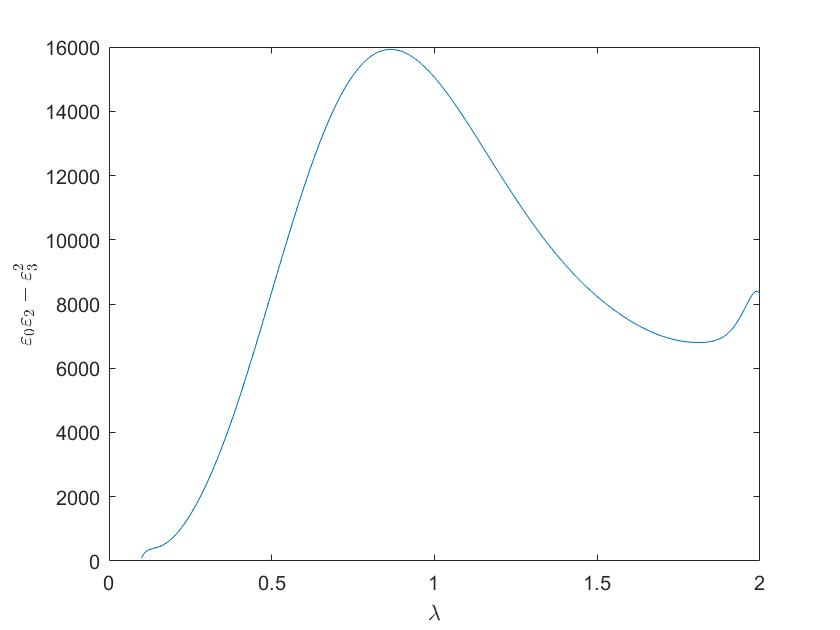}
    \caption{Testing the constraint $\epsilon_0\epsilon_2>\epsilon_3^2$ for the KvBLL skyrmion.}
    \label{fig:epsilon-constraint}
\end{figure}

\noindent\textit{Proof of Proposition \ref{prop:optimal-couplings}}. To minimise \eqref{BMI-squared}, we make use of the following identities.
\begin{align}
\frac{a}{b-x^2} &= \frac{a}{b}\left(1 + \frac{x^2}{b-x^2}\right),\label{i1}\\
\frac{1}{x}(ax+b)^2 &= \frac{1}{x}(ax-b)^2 + 4ab,\label{i2}\\
\frac{1}{x}(ax+b)(cx+d) &= \frac{1}{x}\left(\sqrt{ac}x-\sqrt{bd}\right)^2 + \left(\sqrt{ad}+\sqrt{cd}\right)^2.\label{i3}
\end{align}
First, using \eqref{i1} with $x=x_4$, $a=(1-2x_3)(1-3x_3^2)$, and $b=(1-2x_3)(4x_2-x_3^2)$, and noting that \eqref{YM-constraint-4-relaxed} means $b-x^2\geq0$, gives
\begin{align}
\left(24\pi^2E_{\mathrm{BMI}}\right)^2 \geq \left[\epsilon_0+\sum_{i=1}^3x_i\epsilon_i\right]\left[9 + \frac{(1-3x_3)^2}{4x_2-x_3^2}\right],
\end{align}
with equality if and only if\footnote{This is clearly expected as the energy $E(x)$ doesn't depend on $x_4$ since $\epsilon_4=0$.}
\begin{align}
    x_4=0.\label{x4}
\end{align}
Next, using \eqref{i2} with $x=x_1>0$, $a=\epsilon_1$, and $b=\epsilon_0+x_2\epsilon_2+x_3\epsilon_3$ gives
\begin{align}\label{b2}
\left(24\pi^2E_{\mathrm{BMI}}\right)^2 \geq 4\epsilon_1\left[\epsilon_0+\sum_{i=2}^3x_i\epsilon_i\right]\left[9 + \frac{(1-3x_3)^2}{4x_2-x_3^2}\right],
\end{align}
with equality if and only if \eqref{x4} and\footnote{The condition \eqref{x1} is precisely the virial theorem one obtains by the usual Derrick scaling argument.}
\begin{align}
x_1=\frac{1}{\epsilon_1}\left[\epsilon_0+\sum_{i=2}^3x_i\epsilon_i\right].\label{x1}
\end{align}
Finally, using \eqref{i3} on the right-hand-side of \eqref{b2} with $x=4x_2-x_3^2$, $a=\epsilon_2$, $b=(1-3x_3)^2$, $c=9$, and $d=\epsilon_2x_3^2+4\epsilon_3x_3+4\epsilon_0$, and noting \eqref{YM-constraint-3-relaxed} means $x\geq0$, we obtain
\begin{align}
\left(24\pi^2E_{\mathrm{BMI}}\right)^2 \geq \epsilon_1 \left(\sqrt{\epsilon_2}|1-3x_3| + 3\sqrt{\epsilon_2x_3^2+4\epsilon_3x_3+4\epsilon_0}\right)^2,\label{min124}
\end{align}
with equality subject to \eqref{x4}, \eqref{x1}, and 
\begin{align}
4x_2-x_3^2 = \frac{|1-3x_3|\sqrt{\epsilon_2x_3^2+4\epsilon_3x_3+4\epsilon_0}}{3\sqrt{\epsilon_2}}, \label{x2}
\end{align}
which provides a condition for $x_2$. It remains to optimise the right hand side of \eqref{min124} with respect to $x_3$. The expression inside the brackets is the sum of two convex functions of $x_3$, so has at most one minimum. The derivative of this expression is discontinuous at $x_3=1/3$, and this point will be a local minimum (and hence the global minimum) if and only if the derivative changes sign at this point.  We find that
\begin{align}
\lim_{x_3\to \frac{1}{3}^\pm} \frac{\d\;}{\d x_3}\left(\sqrt{\epsilon_2}|1-3x_3| + 3\sqrt{\epsilon_2x_3^2+4\epsilon_3x_3+4\epsilon_0}\right)= \frac{3(\epsilon_2+6\epsilon_3)}{\sqrt{\epsilon_2+12\epsilon_3+36\epsilon_0}} \pm 3\sqrt{\epsilon_2} 
\end{align}
So this function has a local minimum at $x_3=\frac13$ if and only if
\begin{align}
\sqrt{\epsilon_2} > \frac{\epsilon_2+6\epsilon_3}{\sqrt{\epsilon_2+12\epsilon_3+36\epsilon_0}}.
\end{align}
This condition is equivalent to
\begin{align}
\epsilon_0\epsilon_2> \epsilon_3^2.
\end{align}
Plugging $x_3=1/3$ into \eqref{min124} yields the inequality \eqref{energy-bound-optimal}, and the parameters \eqref{optimal-couplings} follow from this and \eqref{x4}, \eqref{x1}, and \eqref{x2}.\hfill$\square$\\
\subsubsection{Optimal BMIs for KvBLL skyrmions}
Using the data obtained by interpolating table \ref{tab:KvBLL-cal-energy-terms}, we may determine the BMI for the optimal couplings \eqref{optimal-couplings} as a function of the KvBLL scale parameter $\lambda$. We may similarly determine, for fixed $\lambda$, which value of the ultraspherical parameter $\alpha$ optimises $E_{\mathrm{BMI}}$. A comparison of the BMI for the ultraspherical family\footnote{We have coloured with respect to $\log\alpha$ as the colouring is less clear for $\alpha$. This is also well-defined as the ideal ultraspherical parameter is always positive.}, and the optimal couplings is given in figure \ref{fig:BMI-ultraspherical-vs-optimal}.
\begin{figure}[ht]
    \centering
    \includegraphics[width=0.6\linewidth]{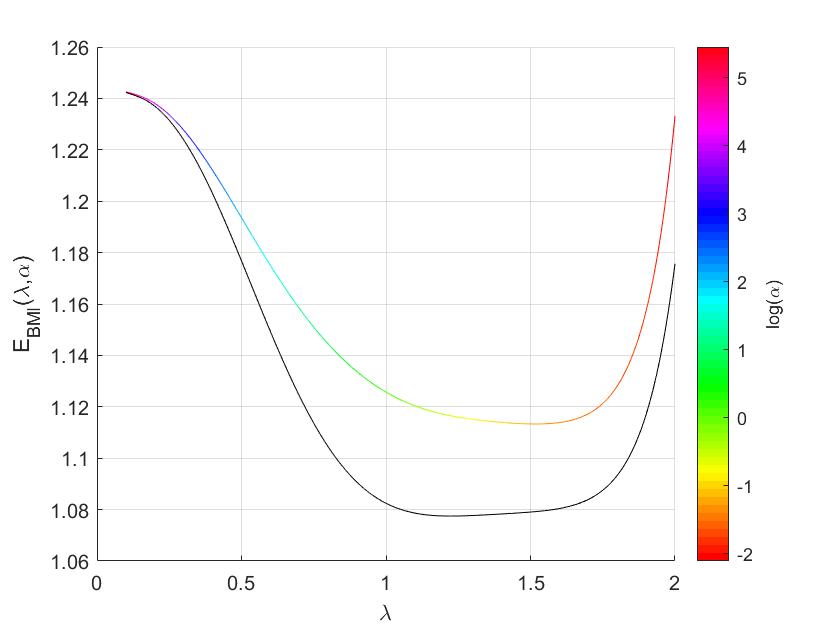}
    \caption{The BMI as a function of $\lambda$ for the ideal ultraspherical parameter (coloured), against the optimal couplings (black).}
    \label{fig:BMI-ultraspherical-vs-optimal}
\end{figure}
A few observations can be made from this. The first thing to notice is that the shape of the two curves are similar, suggesting that the ultraspherical couplings provide a reasonable approximation to an optimal near BPS model. It is important to note that whilst the curves are very close near $\lambda=0$, this is simply an artifact of the fact that near $\lambda=0$, the KvBLL caloron looks like an ordinary instanton on $\R^4$, and can be considered as the weak-coupling limit, i.e. we are considering something which is more like an ungauged skyrmion, and hence there are no preferred couplings. The next observation is that there are versions of the energy \eqref{Gauged-Skyrme-energy-x-vars} for which the KvBLL skyrmion (at certain scales) has a very low BMI. The minimal value obtained for the optimal couplings is $E_{\mathrm{BMI}}=1.078$, for $\lambda=1.223$. The BMI in the ultraspherical model for the same $\lambda$ was obtained at $\alpha=0.55$, with $E_{\mathrm{BMI}}(\alpha)=1.117$, which is approximately $3.7\%$ above the optimum. Both of these values should be considered as being at strong-coupling, and are both significant improvements on the equivalent BMIs in the traditional model.
\subsubsection{Numerical minimisation in the optimal model}
We now consider the numerical minimiser for the optimal couplings calculated in the previous section. In particular,  we consider the caloron corresponding to $\lambda = 1.223$,  where we find from \eqref{optimal-couplings} the optimal set of couplings to be
\begin{align}
x_1 = 1.65084, \quad x_2 = \frac{1}{36}, \quad x_3 = \frac{1}{3}, \quad x_4 = 0, \quad x_5 = \frac{1}{3}, \quad x_6 = 1.
\end{align}
We then use the KvBLL skyrmion considered in the previous section for $\lambda = 1.223$ as an initial condition and numerically minimise to find the local minimum displayed in figure \ref{fig.optimalSkyrmePlot}. This numerical solution has
% energy $E = 108.229$ and 
$E_{\rm BMI}=1.067$. This is even lower than the BMI for the KvBLL skyrmion above, although the two values are extremely close, with the KvBLL BMI only $1\%$ above the true minimum.  This is a very low BMI and no other gauged Skyrme model has come close to such a small value.
% It also compares well to well-known low BMIs in other variants of the Skyrme model, being less than double that of the Skyrme crystal which has an admirable BMI of $E_{\rm BMI}\approx 1.036$ \cite{BattyeSutcliffe1998skyrmecrystal}.
\begin{figure}[ht]
\centering
\includegraphics[width=0.5\linewidth]{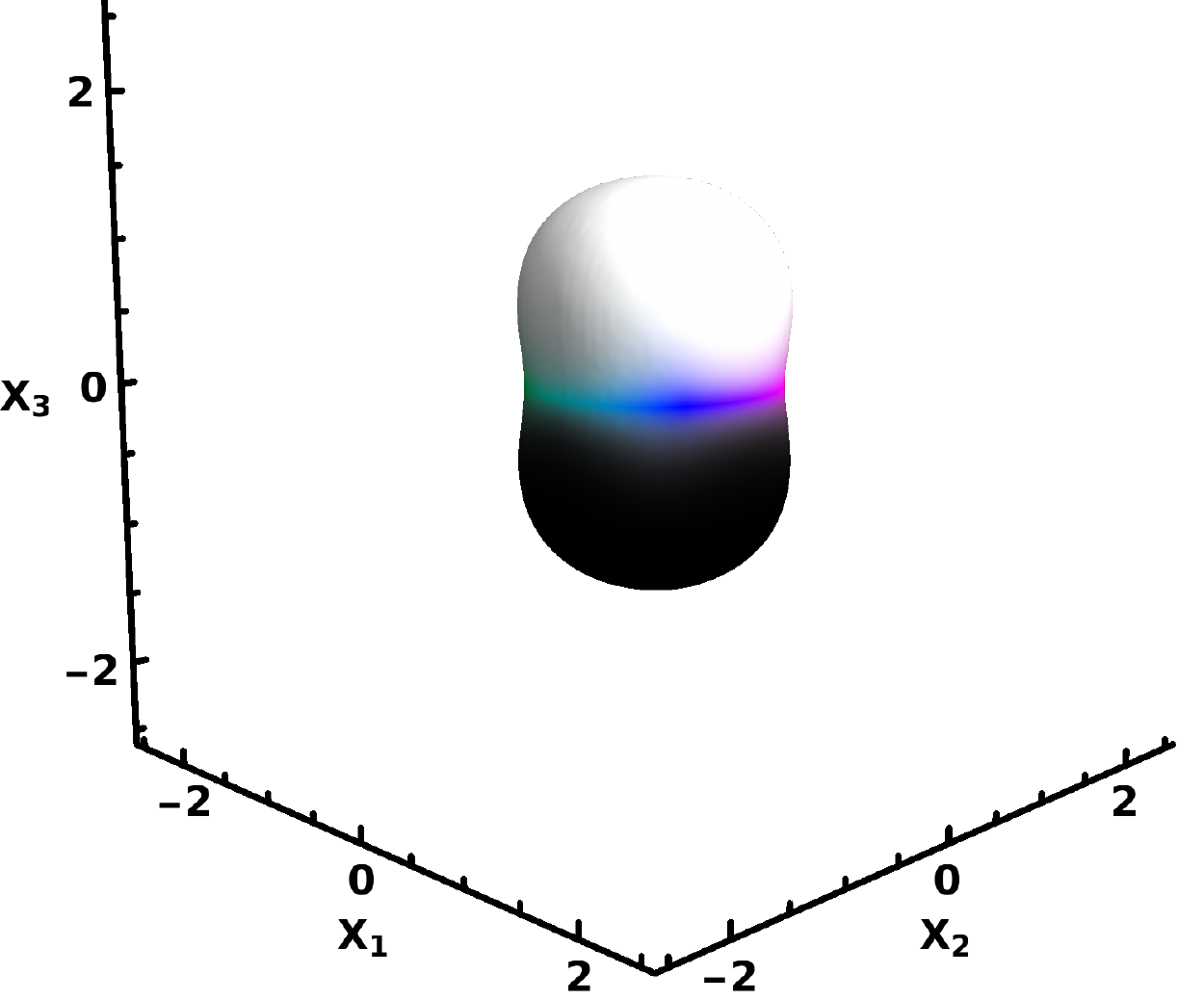}
\caption{Isosurface plot of the energy density for a local minimiser of the gauged Skyrme model. The model used optimal couplings, calculated in the previous section. The plot is coloured by the direction of the pion fields and the isosurface is at the value $\EE = 2.0$. }
\label{fig.optimalSkyrmePlot}
\end{figure}

\section{Dipole moments}\label{sec:dipole}
We saw in the previous section that the skyrmions in the ultraspherical family of gauged Skyrmions that satisfy the boundary condition \eqref{boundary-condition} have low energies and low BMIs.  This boundary condition breaks the gauge group to $\U(1)$.  In this section we explore the electromagnetic properties of these skyrmions.  The interpretation of the KvBLL caloron as a pair of two monopoles \cite{kraanVanBaal1998monopoleconsts} will allow us to estimate the magnetic dipole moment of the nucleon in our model.

\subsection{Abelianisation}
First we explain how to interpret the asymptotic fields of a skyrmion in terms of pions and photons.  Recall from \eqref{boundary-condition} that $U=-{\rm i}\sigma^3$ at $\infty$.  In order to interpret $U$ in terms of pions, we therefore write $U$ in terms of pion fields $\pi_1,\pi_2,\pi_3$ as follows:
\begin{align}
U(\vec{X}) = -{\rm i}\sigma^3\exp({\rm i} \pi_j(\vec{X})\sigma^j).\label{pion-field-skyrme}
\end{align}
With this identification, the $\SU(2)$ gauge group corresponds to the non-diagonal subgroup of $\SU(2)_L\times \SU(2)_R$ generated by $(-{\rm i}\sigma^1,{\rm i}\sigma^1)$, $(-{\rm i}\sigma^2,{\rm i}\sigma^2)$ and $({\rm i}\sigma^3,{\rm i}\sigma^3)$.  This contains a mixture of both vector and axial transformations, but the boundary condition $\pi_j=0$ breaks the gauge symmetry to $\U(1)_V$.

Near spatial infinity, we write the fields $(U,A)$ as
\begin{align}
U &\sim -{\rm i}\sigma^3(\mathbb{1}+{\rm i} \pi_j\sigma^j) =-{\rm i}\sigma^3+\pi_3\mathbb{1}-\pi_+\sigma^- + \pi_-\sigma^+,\label{Asymptotic-Skyrme-field}\\
 A &\sim -\frac{{\rm i} e}{2\hbar}a\sigma^3-{\rm i} W\sigma^--{\rm i}\overline{W}\sigma^+\label{Asymptotic-gauge-field}
\end{align}
with $\pi_\pm=\pi_1\pm{\rm i}\pi_2$, $\pi_3,a,W$ small, and $\sigma^{\pm}=\frac{1}{2}(\sigma^1\pm{\rm i}\sigma^2)$. Thus to leading order, the Dirichlet energy density is
\begin{multline}\label{Dirichlet}
    |L^A|^2\sim2\left(\d\pi_3-{\rm i}\ol{W}\pi_++{\rm i} W\pi_-\right)\wedge\star\left(\d\pi_3-{\rm i}\ol{W}\pi_++{\rm i} W\pi_-\right)\\+2\left(2W+{\rm i}\dfrac{e}{\hbar}a\pi_++\d\pi_+\right)\wedge\star\left(2\ol{W}-{\rm i}\dfrac{e}{\hbar}a\pi_-+\d\pi_-\right).
\end{multline}
The second line of \eqref{Dirichlet} acts as a mass term for $W$ in the energy \eqref{Gauged-Skyrme-energy-x-vars}, forcing $W$ to decay exponentially. Thus, we can ignore $W$ from now on. The remaining terms contribute to the leading order expansion of the energy \eqref{Gauged-Skyrme-energy-x-vars}, which takes the form
\begin{align}
    E\sim\int_{\R^3}2x_1\left(\d\pi_3\wedge\star\d\pi_3+\d_a\pi_+\wedge\star\d_a\pi_-\right)+x_6\dfrac{e^2}{2\hbar^2}\d a\wedge\star\d a,
\end{align}
where $\d_a\pi_{\pm}=\d\pi_\pm\pm{\rm i}\frac{e}{\hbar}a\pi_\pm$. This describes the energy of massless pion fields coupled to an electromagnetic potential $a$. The coupling constants should therefore \cite{MantonSutcliffe2004,PietteTchrakian2000static,AdkinsNappiWitten1983static} be identified with the following physical quantities:\footnote{We use Heaviside--Lorentz units with the speed of light set to $c=1$, so $\varepsilon_0=1$ and the dimensionless fine structure constant is $\alpha=e^2/4\pi\hbar$.}
\begin{align}\label{calibration-asymptotically}
x_1=\frac{F_\pi^2}{16\hbar},\quad x_6=\frac{\hbar^2}{e^2}=\frac{\hbar}{4\pi\alpha}.
\end{align}
In addition, by standard convention \cite{MantonSutcliffe2004}, the coefficient of the Skyrme term in \eqref{Gauged-Skyrme-energy-x-vars} is written
\begin{align}\label{x2-Manton-Sutcliffe}
x_2 = \frac{\hbar}{16g^2}
\end{align}
for a dimensionless constant $g$.
\subsection{The dipole moment of the KvBLL skyrmion}
In a suitable gauge, the KvBLL caloron has asymptotics\footnote{See appendix \ref{appendix:KvBLL} for a more detailed discussion.}
\begin{align}
A_t &\sim {\rm i}\left(\frac{\pi}{2\beta}+\frac{1}{2r_2}-\frac{1}{2r_1}\right)\sigma^3
\end{align}
where $\sim$ denotes equality up to exponentially decaying terms and
\begin{align}
r_1=\sqrt{X_1^2+X_2^2+(X_3-\lambda^2/4)^2},\quad r_2=\sqrt{X_1^2+X_2^2+(X_3+\lambda^2/4)^2}.
\end{align}
It follows from the anti-self-duality equation, and comparison with \eqref{Asymptotic-gauge-field} that the asymptotic magnetic field is
\begin{align}
\star\d a \sim \frac{\hbar}{e}\d\left(\frac{1}{r_1}-\frac{1}{r_2}\right),
\end{align}
which takes the form of two oppositely charged magnetic monopoles. This exhibits a magnetic dipole moment of magnitude
\begin{align}
\mu = \frac{2\pi\hbar\lambda^2}{e}.
\end{align}
As a consequence, the KvBLL skyrmion has a dipole moment, and we seek to compare this with the dipole moments of the neutron and proton.

Traditionally the magnetic dipole moment $\mu$ of a nucleon is expressed in terms of the nuclear magneton $\mu_N=e\hbar/2M_N$, with $M_N$ being the nucleon mass.  So we are interested in the physical quantity
\begin{align}
\frac{\mu}{\mu_N} = \frac{4\pi\lambda^2M_N}{e^2}.
\end{align}
We need to express this in ``Skyrme language''.  In the Adkins--Nappi--Witten calibration \cite{AdkinsNappiWitten1983static},
\begin{align}
\frac{M_N}{F_\pi}=\frac{939MeV}{129MeV}\approx7.28.
\end{align}
Also, using \eqref{calibration-asymptotically}-\eqref{x2-Manton-Sutcliffe} we can express $F_\pi/e^2$ in terms of our couplings:
\begin{align}
\frac{F_\pi}{e^2}=\frac{\hbar^2}{e^2}\frac{F_\pi}{\sqrt{\hbar}}\frac{1}{\hbar^{\frac32}} = \frac{x_6}{16g^3}\sqrt{\frac{x_1}{x_2^3}}.
\end{align}
Using the Adkins--Nappi-Witten value $g=5.45$ we hence obtain
\begin{align}
\frac{\mu}{\mu_N} = \frac{7.28\times \pi}{4\times 5.45^3}x_6\lambda^2\sqrt{\frac{x_1}{x_2^3}}.
\end{align}

Now we estimate the quantities that enter this expression. Since we assume the coupling of the electromagnetic field to the Skyrme field is weak, we can assume that $\beta\gg1$. In line with prior motivations, we shall assume that the couplings are defined by \eqref{YM-couplings-ints} by a family of functions $\varphi_+$, parameterised by $\beta$, such that
\begin{align}
\lim_{\beta\to\infty}\varphi_+(t) = \frac{1}{2}\left(1+\operatorname{erf}\left(\frac{t}{\sqrt{2}}\right)\right).
\end{align}
For example, we know that such a family is given by the ultraspherical kink \eqref{hypergeometric-fct} when $\alpha=\beta^2/8$. Thus we may use the limiting values
\begin{align}
\begin{aligned}
x_1 &\approx \int_{-\infty}^\infty \frac{1}{2\pi}e^{-t^2}\;\d t = \frac{1}{2\sqrt{\pi}},\\
x_2 &\approx \int_{-\infty}^\infty \left(\frac{1}{4}-\frac{1}{4}\mathrm{erf}\left(\frac{t}{\sqrt{2}}\right)\right)^2\;\d t \approx 0.0990.
\end{aligned}
\end{align}
We also need to estimate the scale parameter $\lambda$. In the weak coupling limit the optimal caloron is expected to be instanton-like.  One can show that \cite{LeeLu1998} in this limit the caloron resembles a  CF'tH instanton \cite{CorriganFairlie1977} with pre-potential
\begin{align}
\phi = 1 + \frac{\lambda^2}{\mu_0|x|^2}.
\end{align}
Working in units where $x_1=x_2=\frac12$, Atiyah and Manton \cite{AtiyahManton1989} estimated the optimum scale of this instanton, when approximating the ordinary $1$-skyrmion, to be $\lambda^2/\mu_0\approx 2.109$.  So we know that
\begin{align}
\frac{\lambda^2}{\mu_0}\frac{x_1}{x_2}\approx 2.109\implies \lambda^2 \approx \frac{2.109\times 2\pi}{\beta}\frac{x_2}{x_1}.
\end{align}
Since we also know that $x_6=\beta$, we finally obtain
\begin{align}
\frac{\mu}{\mu_N} \approx \frac{7.28\times 2.109\times \pi^2}{2\times 5.45^3\times \sqrt{x_1x_2}} \approx \frac{7.28\times 2.109\times \pi^{9/4}}{5.45^3 \times\sqrt{2\times 0.0990}} \approx 2.80.
\end{align}
This compares well with the experimental values of 2.79 for the proton and 1.91 for the neutron. This suggests that the KvBLL caloron, and the gauged skyrmions it approximates, provide a realistic approximation (within the realms of general Skyrme model accuracy) to real nuclei.

Although it didn't enter the calculation, we can estimate the period $\beta$.  The fine structure constant\footnote{This notation should not be confused with the ultraspherical parameter, which is also called $\alpha$.} $\alpha$ should be close to its physical value of $\frac{1}{137}$, so
\begin{align}
\label{cc}
\frac{x_6}{x_2}=\frac{4g^2}{\pi\alpha}\approx \frac{4\times 5.45^2\times 137}{\pi}.
\end{align}
Using $x_6=\beta$ and $x_2\approx 0.0990$ we obtain
\begin{align}
\beta \approx \frac{4\times 5.45^2\times 137\times 0.0990}{\pi} \approx 513.
\end{align}
To reproduce this in the ultraspherical family, we would set the ultraspherical parameter as $\alpha=\beta^2/8\approx32.9\times 10^2$.

Although we haven't done so, one could go on to compute further properties of nucleons by quantising the gauged 1-skyrmion, following the methods of \cite{AdkinsNappiWitten1983static}.  However, there are two ways in which the calculation for gauged skyrmions would differ from \cite{AdkinsNappiWitten1983static}.  The first is that in the case of gauged skyrmions, both the vacuum and the soliton possess less symmetry than in the ungauged case.  In the gauged case the boundary condition breaks the group $\SU(2)_J\times \SU(2)_I$ of rotations and isorotations down to $\SU(2)_J\times \U(1)_I$. But, as the soliton is invariant only under a $\U(1)$ subgroup, the configuration space on which to quantise is $(\SU(2)_J\times \U(1)_I)/\U(1)\cong\SU(2)$, as in the ungauged case.

The second difference for gauged skyrmions is that, in order to reproduce the Gell-Mann--Nishijima relations (and several other important relations), it is necessary to include a Wess--Zumino--Witten term in the Skyrme lagrangian. Unlike in the ungauged case, the Wess--Zumino--Witten term is non-vanishing in the gauged Skyrme model. In our holographic picture this would correspond to studying calorons in the presence of a Chern--Simons term. The effect of Chern--Simons terms on calorons have not been studied, but its effect on instantons on $\R^4$ has been \cite{CollieTong2008instantons}.  It would be interesting to see whether the methods of \cite{CollieTong2008instantons} can be applied to calorons, how they affect the dynamics of calorons (studied earlier in \cite{LeeYi2003}), and what the implications for gauged skyrmions would be.

\section{Concluding remarks}\label{sec:conclusion}

We have explored choices of coupling constants in the gauged Skyrme model that result in energies close to the topological energy bound -- in other words, with low BMI \eqref{BMI}.  We did this in two ways.  First, using a family of coupling constants based on an holographic construction, we obtained a 1-soliton with energy 10\% above the bound derived in section \ref{sec:energy-bounds}. Second, by optimising coupling constants to minimise the BMI of the KvBLL skyrmion, we obtained an energy that is just 7\% above the topological energy bound.  These results improve on previous results of 82\% in a gauged Skyrme model \cite{PietteTchrakian2000static} and 23\% in the standard ungauged Skyrme model. As a consequence, our model will have much lower classical binding energies than more conventional models.

Our results were obtained in an $\SU(2)$-gauged Skyrme model, which is physically unrealistic.  However, our choice of boundary condition broke this gauge symmetry to $\U(1)$, so our results give a good indication of what could be expected in the more realistic $\U(1)$-gauged Skyrme model.  The low values for BMI that we quoted above were obtained with a coupling which is much stronger than the real-life electromagnetic coupling.  However, our strong coupling results still give insight into how the coefficients of the gauged Skyrme model should be chosen at weak coupling.  An example of this is in section \ref{sec:dipole}, where we showed that our model gives a realistic value for the nucleon magnetic dipole moment at weak coupling.

The results reported here are (to the best of our knowledge) the first systematic exploration of the parameter space for gauged Skyrme models. Previous studies have used a restrictive set of parameters, for which most of the terms in the energy \eqref{Gauged-Skyrme-energy} are absent. We have shown that the additional terms in \eqref{Gauged-Skyrme-energy} allow for much lower BMIs, and hence lower binding energies.

\subsection*{Acknowledgements}
JC is grateful to Lancaster University, UK, for allowing extended access to MATLAB, through which many of the results of this work were calculated. TW was supported by the University of Leeds as an Academic Development Fellow throughout the duration of this project, and by the UK Engineering and Physical Sciences Research Council
through grant EP/P024688/1.  Some of the numerical simulations were run using the Soliton Solver library
developed by TW at the University of Leeds. This work was undertaken on ARC4, part of the High Performance Computing facilities at the University of Leeds, UK.

\appendix
\section{KvBLL calorons}\label{appendix:KvBLL}
In this appendix, we give a sufficient review of the KvBLL calorons of \cite{KraanVanBaal1998,LeeLu1998} and the Nahm transform for calorons \cite{Nyethesis,CharbonneauHurtubise2007nahm.tfm} for the purpose of clarifying parameter conventions used in the main body of this paper.
\subsection{The Nahm transform for calorons}
The Nahm data for an $\SU(2)$ caloron of period $\beta=\frac{2\pi}{\mu_0}$ is defined on the reciprocal circle $S^1=\bigslant{\R}{\mu_0\Z}$, covered by two intervals $I_1=[-\mu,\mu]$ and $I_2=[\mu,\mu_0-\mu]$ of lengths $\nu_1$ and $\nu_2$ respectively. The data consists of the main \textit{bulk data} called \textit{Nahm matrices}, and \textit{fundamental data}, determined by \textit{matching conditions} on the boundaries of the intervals, which take varying forms dependent on the values of the monopole charges $m_1$ and $m_2$. For our purposes, we are only interested in the case where $m_1=m_2=k$, and thus we shall not describe the other scenarios. For this, we refer the reader to \cite{Nyethesis,CharbonneauHurtubise2007nahm.tfm}.
\paragraph{Bulk data.} The \textit{Nahm matrices} are anti-hermitian matrix-valued functions
\begin{align}
    T_p^\alpha:I_p\lto \mathfrak{u}(k),\quad p=1,2,\quad \alpha=0,1,2,3,
\end{align}
analytic on the interior of $I_p$. These are to satisfy \textit{Nahm's equations}
\begin{align}\label{Nahmseqn}
       \dfrac{\d T_p^j}{\d s}+[T_p^0,T_p^j]+\dfrac{1}{2}\epsilon_{jkl}[T_p^k,T_p^l]=0,\quad\text{for }j=1,2,3.
\end{align}
One can always use a periodic gauge transformation (that is functions $g_p:I_p\lto U(k)$ with $g_1(\mu)=g_2(\mu)$ and $g_1(-\mu)=g_2(\mu_0-\mu)$) to fix $T_p^0$ as a constant.
\paragraph{Matching data.} At the boundaries of the intervals, the Nahm matrices are to satisfy the matching conditions
    \begin{align}\label{(k,k)-match}
    \begin{array}{c}
        \AA_1(\zeta,\mu)-\AA_2(\zeta,\mu)=\left(u_+-w_+\zeta\right)\left(w_+^\dagger+u_+^\dagger\zeta\right),\\
        \AA_2(\zeta,\mu_0-\mu)-\AA_1(\zeta,-\mu)=\left(u_--w_-\zeta\right)\left(w_-^\dagger+u_-^\dagger\zeta\right),\end{array}
    \end{align}
    where
    \begin{align}
       \AA_p(s,\zeta)=T_p^2(s)+{\rm i} T_p^3(s)+2{\rm i} T_p^1(s)\zeta+(T_p^2(s)-{\rm i} T_p^3(s))\zeta^2,\quad s\in I_p,\:\zeta\in\C.\label{spectral-fct-A}
    \end{align}
and  $(u_\pm,w_\pm)\in(\C^k)^2\setminus\{(0,0)\}$ are the \textit{matching data}.
\paragraph{The Nahm transform.} To obtain a caloron from its associated Nahm data, we first define the following shorthands. For the matrices $T_p^\alpha\in\mathfrak{u}(k)$, $\alpha=0,1,2,3$, and a point $(t,\vec{X})\equiv(t,X_1,X_2,X_3)\in S^1\times\R^3$, let
\begin{align}
    T_p=T_p^0\otimes\mathbb{1}_2-{\rm i}\sum_{j=1}^3T_p^j\otimes\sigma^j,\quad\text{and}\quad X_p=-{\rm i} t\mathbb{1}_{k}\otimes\mathbb{1}_2-\sum_{j=1}^3X_j\mathbb{1}_{k}\otimes\sigma^j.
\end{align}
These define a Dirac operator on each interval $\slashed{D}_p:C^\infty(I_p,\C^{k}\otimes\C^2)\lto C^\infty(I_p,\C^{k}\otimes\C^2)$, for which, we require functions $\psi_p^q:I_p\lto\C^{k}\otimes\C^2$ in $\ker\slashed{D}$, that is, satisfying
\begin{align}\label{tfm-eqn-m=0}
    \slashed{D}_p(\psi_p^q)=-\left(\dfrac{\d}{\d s}+T_p+X_p\right)\psi_p^q=0,
\end{align}
for $p,q=1,2$. These functions are additionally expected to satisfy matching conditions at $\pm\mu$; letting $v_\pm\in\C^k\otimes\C^2$ be defined by the vectors $(u_\pm,w_\pm)$ as
    \begin{align}\label{v-pm}
        v_\pm=\dfrac{1}{\sqrt{2}}\begin{pmatrix}w_\pm-u_\pm\\
        w_\pm+u_\pm\end{pmatrix},
    \end{align}
we require a pair $\vec{\zeta}^q=(\zeta_+^q,\zeta_-^q)$ for $q=1,2$ satisfying
\begin{align}\label{match-tfm}
    \psi_2^q(\mu)-\psi_1^q(\mu)=v_+\zeta_+^q,\quad \psi_1^q(-\mu)-\psi_2^q(\mu_0-\mu)=v_-\zeta_-^q. 
\end{align}
The functions $\psi_p$ and their corresponding matching data form a local frame $\{e^j\}$ for the fibre over $(t,\vec{X})$ of a rank $2$ vector bundle $V\to S^1\times\R^3$. Each fibre is endowed with a natural hermitian inner product, on which we additionally require the frame to be orthonormal. This is understood via the condition
\begin{align}\label{Normal-(k,k)}
    \ol{\zeta_+^a}\zeta_+^b+\ol{\zeta_-^a}\zeta_-^b+\sum_{p=1}^2\int_{I_p}{\psi_p^a}(s)^\dagger{\psi_p^b}(s)\;\d s=\delta^{ab}.
\end{align}
From this, the caloron is defined as the connection induced from the trivial connection on the infinite rank bundle $S^1\times\R^3\times C^\infty(S^1_{\mu_0},\C^k\otimes\C^2)$, that is the connection whose matrix components in $\su(2)$ are the $1$-forms
\begin{align}
    \wt{A}^{ij}=\langle e^i,\d e^j\rangle.
\end{align}
This may be understood more explicitly in terms of the data $(\psi_p,\zeta_p)$ as follows. By forming the matrices
\begin{align}\label{matrix-version-tfm}
    S_\pm=\begin{pmatrix}
    \zeta_\pm^1&\zeta_\pm^2\end{pmatrix},\quad\psi_p(s)=\begin{pmatrix}
    \psi_p^1(s)&\psi_p^2(s)\end{pmatrix},
\end{align}
the connection matrix $\wt{A}_\mu:S^1\times\R^3\lto\su(2)$ may thus be conceived as
\begin{align}\label{GF-nahm}
    \wt{A}_\mu=S_+^\dagger\bdy_\mu S_++S_-^\dagger\bdy_\mu S_-+\sum_{p=1}^2\int_{I_p}\psi_p(s)^\dagger\bdy_\mu\psi_p(s)\;\d s.
\end{align}
\subsection{Charge \texorpdfstring{$(1,1)$}{(1,1)} Nahm data}
In the case $k=1$, Nahm's equation \eqref{Nahmseqn} implies that the Nahm matrices are constant. The isometries of $S^1\times\R^3$ act in the obvious way on the tuples $(T_p^0,T_p^j)$, and they can be exploited to fix the Nahm matrices as
\begin{align}\label{KvBLL-Nahm-matrices}
    T_1^\alpha={\rm i}\dfrac{\lambda^2}{4}\delta^{\alpha 3},\quad T_2^\alpha=-{\rm i}\dfrac{\lambda^2}{4}\delta^{\alpha 3},
\end{align}
for $\alpha=0,1,2,3$. Up to gauge transformations and a relative phase, the corresponding matching data are given by
\begin{align}\label{KvBLL-matching}
    u_+=-{\rm i}\frac{\lambda}{\sqrt{2}},\quad w_+={\rm i}\frac{\lambda}{\sqrt{2}},\quad u_-={\rm i}\frac{\lambda}{\sqrt{2}},\quad w_-={\rm i}\frac{\lambda}{\sqrt{2}},
\end{align}
where $\lambda>0$. This parameterisation defines the most symmetric examples of KvBLL calorons, with axial symmetry around the $X_3$-axis in $\R^3$, and all other $(1,1)$ calorons may be obtained from this via isometries of the caloron moduli space. For a general set of Nahm matrices, the gauge-invariant quantities given by ${\rm i}\tr(T_p^\alpha)/k$ are interpreted as the \textit{locations} of the constituent monopoles in $S^1\times\R^3$. For the data \eqref{KvBLL-Nahm-matrices}, the locations are thus $(0,0,0,\pm\lambda^2/4)$. When the masses $\nu_1$ and $\nu_2$ are equal, there are additional symmetries that can be understood by switching the constituent monopoles. Explicitly, the data is invariant up to gauge transformations under the isometries $(t,X_1,X_2,X_3)\mapsto(-t,X_1,X_2,-X_3)$ and $(t,X_1,X_2,X_3)\mapsto(t,-X_1,X_2,-X_3)$, and the switching of the constituent monopoles is realised by combining these isometries with a large gauge transformation called \textit{the rotation map} \cite{cork2018symmrot}, which acts on the Nahm data by rotating the circle on which it is defined (which here amounts to $T_1\leftrightarrow T_2$), along with a $\U(1)$ phase, which acts on the matching data only as the action $(u_\pm,w_\pm)\mapsto e^{\pm{\rm i}\theta}(u_\pm,w_\pm)$.
\subsection{Applying the Nahm transform}
In this section we shall reproduce the KvBLL calorons from their Nahm data, that is, we shall apply the Nahm transform explicitly. This was also reviewed recently in \cite{KatoNakamulaTakesue2018magnetically} using a different, but equivalent formulation of the Nahm transform. Our calculation is most similar to that of \cite{LeeLu1998}. Starting with the matching data \eqref{KvBLL-matching}, we may form the vectors $v_\pm$ as in \eqref{v-pm}, which due to our choice of parameterisation, take the simple form
\begin{align}
    v_+={\rm i}\begin{pmatrix}
    \lambda\\
    0\end{pmatrix},\quad v_-={\rm i}\begin{pmatrix}0\\\lambda\end{pmatrix}.
\end{align}
From this, and the matrices \eqref{KvBLL-Nahm-matrices}, the Dirac equations \eqref{tfm-eqn-m=0} reduce to
\begin{align}
    \left(-\dfrac{\d}{\d s}+{\rm i} t+(\vec{X}-\frac{\lambda^2}{4}\vec{e}_3)\cdot\vec{\sigma}\right)\psi_1^q&=0,\label{Nahm-tfm-k=1_I_1}\\
    \left(-\dfrac{\d}{\d s}+{\rm i} t+(\vec{X}+\frac{\lambda^2}{2}\vec{e}_3)\cdot\vec{\sigma}\right)\psi_2^q&=0,\label{Nahm-tfm-k=1_I_2}
\end{align}
where $\psi_p^q:I_p\lto\C^2$. We may straightforwardly solve (\ref{Nahm-tfm-k=1_I_1})-(\ref{Nahm-tfm-k=1_I_2}) to give
\begin{align}
    \psi_1^q(s)&=\varphi_1(s)V_1^q=e^{{\rm i} ts}\sqrt{N_1}\exp\left(\vec{y}_1\cdot\vec{\sigma}s\right)V_1^q,\label{Soln-Nahm-tfm-k=1_I_1}\\
    \psi_2^q(s)&=\varphi_2(s)V_2^q=e^{{\rm i} t(s-\frac{\mu_0}{2})}\sqrt{N_2}\exp\left(\vec{y}_2\cdot\vec{\sigma}\left(s-\dfrac{\mu_0}{2}\right)\right)V_2^q,\label{Soln-Nahm-tfm-k=1_I_2}
\end{align}
where $\vec{y}_1=\vec{X}-\lambda^2/4\vec{e}_3$, $\vec{y}_2=\vec{X}+\lambda^2/4\vec{e}_3$, $V_p^q\in C^\infty(S^1\times\R^3,\C^2)$, and $N_p$ denote the normalisation factors
\begin{align}
    N_1=\dfrac{r_1}{\sinh(\nu_1 r_1)},\quad N_2=\dfrac{r_2}{\sinh(\nu_2r_2)},
\end{align}
with $r_1=|\vec{y}_1|$, $r_2=|\vec{y}_2|$, and $\nu_1=2\mu$ and $\nu_2=\mu_0-2\mu$ denoting the constituent monopole masses. These coordinate changes and additional factors have been chosen without loss of generality so that the orthonormalisation condition \eqref{Normal-(k,k)} is exactly given by
\begin{align}\label{Normalisation-condition-simplified}
    S_+^\dagger S_++S_-^\dagger S_-+V_1^\dagger V_1+V_2^\dagger V_2=\mathbb{1},
\end{align}
where $S_\pm$ is as in \eqref{matrix-version-tfm}, and $V_p=\begin{pmatrix}V_p^1&V_p^2\end{pmatrix}$. This normalisation also means we may calculate from (\ref{GF-nahm}) the caloron gauge field as
\begin{align}\label{KvBLL-cal-GF}
    \wt{A}_\mu=S_+^\dagger\bdy_\mu S_++S_-^\dagger\bdy_\mu S_-+V_1^\dagger\left(\bdy_\mu+\MM_\mu(\nu_1,\vec{y}_1)\right)V_1+V_2^\dagger\left(\bdy_\mu+\MM_\mu(\nu_2,\vec{y}_2)\right)V_2,
\end{align}
where $\MM=\MM_t\;\d t+\vec{\MM}\cdot \d\vec{X}$ is found in each term by integrating $\varphi_p^\dagger\bdy_\mu\varphi_p$ over $I_p$. This turns out to be the BPS monopole with components
\begin{align}
   \MM_t(\nu,\vec{z})&={\rm i}\left(\frac{\nu}{2}\coth(\nu R)-\dfrac{1}{2R}\right)\dfrac{\vec{z}\cdot\vec{\sigma}}{R},\label{higgs-field}\\
   \vec{\MM}(\nu,\vec{z})&=\dfrac{{\rm i}}{2}\left(\dfrac{\nu R}{\sinh(\nu R)}-1\right)\dfrac{\vec{z}\times\vec{\sigma}}{R^2},
\end{align}
where $R=|\vec{z}|$. It remains to calculate $S_\pm$ and $V_p$ from \eqref{match-tfm} and \eqref{Normalisation-condition-simplified}. The matching conditions \eqref{match-tfm} read
\begin{align}
    \varphi_2(\mu)V_2-\varphi_1(\mu)V_1&={\rm i}\lambda P_+ S,\\
    \varphi_1(-\mu)V_1-\varphi_2(\mu_0-\mu)V_2&={\rm i}\lambda P_-S,
\end{align}
where
\begin{align}
    S=\begin{pmatrix}S_+\\
    S_-\end{pmatrix},\quad P_\pm=\dfrac{1}{2}(\mathbb{1}\pm\sigma^3).
\end{align}
These are solved uniquely by
\begin{align}
    V_1&=\dfrac{{\rm i}\lambda}{\sqrt{N_1}}\dfrac{\BB_1^\dagger}{\BB}\left(e^{{\rm i}\frac{\nu_2}{2}t}\exp\left(\frac{\nu_2}{2}\vec{y}_2\cdot\vec{\sigma}\right)P_++e^{-{\rm i}\frac{\nu_2}{2}t}\exp\left(-\frac{\nu_2}{2}\vec{y}_2\cdot\vec{\sigma}\right)P_-\right)S,\label{V_1-explicit}\\
    V_2&=\dfrac{{\rm i}\lambda}{\sqrt{N_2}}\dfrac{\BB_2^\dagger}{\BB}\left(e^{-{\rm i}\frac{\nu_1}{2}t}\exp\left(-\frac{\nu_1}{2}\vec{y}_1\cdot\vec{\sigma}\right)P_++e^{{\rm i}\frac{\nu_1}{2}t}\exp\left(\frac{\nu_1}{2}\vec{y}_1\cdot\vec{\sigma}\right)P_-\right)S,\label{V_2-explicit}
\end{align}
where
\begin{align}
    \BB_1&=e^{-{\rm i}\frac{\mu_0}{2}t}\exp\left(-\dfrac{\nu_2}{2}\vec{y}_2\cdot\vec{\sigma}\right)\exp\left(-\dfrac{\nu_1}{2}\vec{y}_1\cdot\vec{\sigma}\right)-e^{{\rm i}\frac{\mu_0}{2}t}\exp\left(\dfrac{\nu_2}{2}\vec{y}_2\cdot\vec{\sigma}\right)\exp\left(\dfrac{\nu_1}{2}\vec{y}_1\cdot\vec{\sigma}\right),\label{B1}\\
    \BB_2&=e^{-{\rm i}\frac{\mu_0}{2}t}\exp\left(-\dfrac{\nu_1}{2}\vec{y}_1\cdot\vec{\sigma}\right)\exp\left(-\dfrac{\nu_2}{2}\vec{y}_2\cdot\vec{\sigma}\right)-e^{{\rm i}\frac{\mu_0}{2}t}\exp\left(\dfrac{\nu_1}{2}\vec{y}_1\cdot\vec{\sigma}\right)\exp\left(\dfrac{\nu_2}{2}\vec{y}_2\cdot\vec{\sigma}\right),\label{B2}
\end{align}
and\footnote{The choice of index 1 or 2 is arbitrary since they are symmetric.} $\BB=\tr(\BB_1\BB_1^\dagger)/2=\tr(\BB_2\BB_2^\dagger)/2$. It is straightforward to calculate
\begin{align}\label{B-tr}
    \BB=2\left(\cosh(\nu_1 r_1)\cosh(\nu_2r_2)+\dfrac{\vec{y}_1\cdot\vec{y}_2}{r_1r_2}\sinh(\nu_1r_1)\sinh(\nu_2r_2)-\cos(\mu_0t)\right).
\end{align}
Equations \eqref{V_1-explicit}-\eqref{V_2-explicit} show that $V_p$ may be fully determined once we know the matrix $S$. This may be found by settling the normalisation condition (\ref{Normalisation-condition-simplified}). We find
\begin{align}
    V_1^\dagger V_1
    % &=\dfrac{N_2}{N_1}\dfrac{\lambda^2}{\BB}\left(S_+^\dagger(\varphi_2(\mu)^{-1})^\dagger+S_-^\dagger(\varphi_2(\mu_0-\mu)^{-1})^\dagger\right)\left(\varphi_2(\mu)^{-1}S_++\varphi_2(\mu_0-\mu)^{-1}S_-\right)\\
    % &=\dfrac{1}{N_1}\dfrac{\lambda^2}{\BB}\left(S_+^\dagger\exp(\nu_2\vec{y}_2\cdot\vec{\sigma})S_++S_-^\dagger\exp(-\nu_2\vec{y}_2\cdot\vec{\sigma})S_-\right)\\
    &=\dfrac{1}{N_1}\dfrac{\lambda^2}{\BB}\left(\cosh(\nu_2r_2)+\sinh(\nu_2r_2)\dfrac{y_2^3}{r_2}\right)S^\dagger S,\label{V1V1}\\
    V_2^\dagger V_2
    % &=\dfrac{N_1}{N_2}\dfrac{\lambda^2}{\BB}\left(S_+^\dagger(\varphi_1(\mu)^{-1})^\dagger+S_-^\dagger(\varphi_1(-\mu)^{-1})^\dagger\right)\left(\varphi_1(\mu)^{-1}S_++\varphi_1(-\mu)^{-1}S_-\right)\\
    % &=\dfrac{1}{N_2}\dfrac{\lambda^2}{\BB}\left(S_+^\dagger\exp(-\nu_1\vec{y}_1\cdot\vec{\sigma})S_++S_-^\dagger\exp(\nu_1\vec{y}_1\cdot\vec{\sigma})S_-\right)\\
    &=\dfrac{1}{N_2}\dfrac{\lambda^2}{\BB}\left(\cosh(\nu_1 r_1)-\sinh(\nu_1 r_1)\dfrac{y_1^3}{r_1}\right)S^\dagger S,\label{V2V2}
\end{align}
so that (\ref{Normalisation-condition-simplified}) reads
\begin{align}\label{S-condition}
    S^\dagger S=\dfrac{1}{\NN}\mathbb{1},
\end{align}
where
\begin{align*}
    \NN=1+\frac{\lambda^2}{\BB}\left(\dfrac{1}{N_1}\left(\cosh(\nu_2r_2)+\sinh(\nu_2r_2)\dfrac{y_2^3}{r_2}\right)+\dfrac{1}{N_2}\left(\cosh(\nu_1 r_1)-\sinh(\nu_1 r_1)\dfrac{y_1^3}{r_1}\right)\right).
\end{align*}
Equation (\ref{S-condition}) says we may fix $S$ to be of the form $S=\UU/{\sqrt{\NN}}$, where $\UU:S^1\times\R^3\lto \U(2)$.  Note that from the formula \eqref{GF-nahm} and the normalisation condition \eqref{Normalisation-condition-simplified}, we can see that the action of a gauge transformation
\begin{align}
    \wt{A}_\mu\mapsto g^{-1}\wt{A}_\mu g+g^{-1}\bdy_\mu g
\end{align}
for $g:S^1\times\R^3\lto\U(2)$ is equivalent to the action
\begin{align}
    S\mapsto Sg.
\end{align}
Thus, $\UU$ is determined uniquely up to right-multiplication by gauge transformations. If we fix the gauge $\UU=\exp\left({\rm i}\mu t\sigma^3\right)$, then the boundary condition \eqref{caloron-bdy} is upheld.
\subsection{The tail of the KvBLL caloron}
Using the formula \eqref{KvBLL-cal-GF}, it is possible to calculate the components $\wt{A}_\mu$ of the KvBLL caloron explicitly, and this is most straightforwardly done in the equal mass case $\nu_1=\nu_2$. Using this, we can easily consider the asymptotic value of the connection. In particular, in the region where $r_1,r_2\gg \beta$, that is, at a scale much larger than the period, the component $\wt{A}_t$ is approximately abelian, and is given by the asymptotic value
\begin{align}\label{tail-KvBLL-simple}
    \wt{A}_t={\rm i}\left(\dfrac{\mu_0}{4}+\dfrac{1}{2r_2}-\dfrac{1}{2r_1}\right)\sigma^3+O(e^{-\mu_0r_p/2}).
\end{align}
Remarkably, this formula can be obtained in a more straightforward manner, simply from the Nahm data. It is well-known that the curves defined by $\CC_p(\eta,\zeta)=\det\left(\eta+\AA_p(\zeta)\right)=0$, where $\AA_p$ are as in \eqref{spectral-fct-A}, are independent of $s\in I_p$. These are known as the \textit{spectral curves} of the caloron. Letting $(X_1,X_2,X_3)\in\R^3$, and for $\zeta\in\C$, define $x(\zeta)=-(X_3-{\rm i} X_2)-2X_1\zeta+(X_3+{\rm i} X_2)\zeta^2$. From this, and the curves $\CC_p$, we define the objects
\begin{align}\label{tail-cal}
    \mathfrak{V}_p=\dfrac{{\rm i}}{2\pi}\oint_{\gamma}\left.\dfrac{\bdy_\eta \CC_p}{\CC_p}\right|_{\eta=x(\zeta)}\;\d\zeta,
\end{align}
where $\gamma$ is a contour containing only one half of the poles of the integrand. These are known as the \textit{tails}  of the caloron's constituent monopoles, in analog to the tail of a BPS monopole \cite{HarlandNogradi2016charge,hurtubise1985asymptotic}. In the case of BPS monopoles there is only one tail $\mathfrak{V}$, and Hurtubise's theorem \cite{hurtubise1985asymptotic} says that asymptotically, the Higgs field satisfies
\begin{align}\label{tail-monopole-higgs}
    |\Phi|\sim \nu+\mathfrak{V},
\end{align}
where $\nu$ is the monopole mass. It is currently unknown whether a similar result holds for calorons, and therefore it is unknown how best to define the tails \eqref{tail-cal}.

It is straightforward enough to compute from \eqref{KvBLL-Nahm-matrices} that the spectral curves for our parameterisation of the KvBLL calorons are given by the functions
\begin{align}
    \CC_{1}(\eta,\zeta)=\eta-\dfrac{\lambda^2}{4}\left(\zeta^2-1\right),\quad\CC_{2}(\eta,\zeta)=\eta+\dfrac{\lambda^2}{4}\left(\zeta^2-1\right).
\end{align}
We thus have
\begin{align}
    \GG_1:=\left.\dfrac{\bdy_\eta \CC_1}{\CC_1}\right|_{\eta=x(\zeta)}(\zeta)&=\left(x(\zeta)-\frac{\lambda^2}{4}(\zeta^2-1)\right)^{-1},\\
    \GG_2:=\left.\dfrac{\bdy_\eta \CC_2}{\CC_2}\right|_{\eta=x(\zeta)}(\zeta)&=\left(x(\zeta)+\frac{\lambda^2}{4}(\zeta^2-1)\right)^{-1}.
\end{align}
Away from the rays $(X_1,X_2,X_3)=(a,0,\mp\lambda^2/4)$, $a\neq0$, these have simple poles at $\zeta_1^{\pm}$ and $\zeta_{2}^\pm$ respectively, where
\begin{align}
\begin{aligned}
    \zeta_1^\pm&=\dfrac{X_1\pm r_{1}}{X_3+{\rm i} X_2-\frac{\lambda^2}{4}},&
    \zeta_2^\pm&=\dfrac{X_1\pm r_{2}}{X_3+{\rm i} X_2+\frac{\lambda^2}{4}}.
\end{aligned}
\end{align}
We also have
\begin{align}
    \mathrm{Res}\left(\GG_p,\zeta_p^\pm\right)=\pm\dfrac{1}{2r_p}.
\end{align}
By construction, $\gamma$ only contains one of the poles of the integrand of \eqref{tail-cal} in each case. If it is $\zeta_p^+$ in each case, then we find that the tail \eqref{tail-KvBLL-simple} of the KvBLL caloron can be straightforwardly reproduced by the formula
\begin{align}
    \wt{A}_t={\rm i}\left(\dfrac{\mu_0}{4}+\mathfrak{V}_1-\mathfrak{V}_2\right)\sigma^3+\text{exponentially decaying terms},
\end{align}
with the order of the constituents $1$ and $2$ reversed if the opposite poles are inside $\gamma$. This sort of formula is the candidate for the analogue of Hurtubise's formula \eqref{tail-monopole-higgs} for the tail of a monopole, and was first considered and verified for other higher charge calorons in \cite{bruckmannvanBaalNog2004higher}. The opposite signs of the individual tails accounts for the fact that the monopole constituents have opposite charges. The usefulness of the asymptotic formulae for deriving magnetic moments of caloron generated gauged skyrmions presented in this paper therefore provides further motivation for resolving this important unsolved problem.
\newpage
\bibliographystyle{unsrt}
\bibliography{refs}
\end{document}